%% file: mainfile.tex
\documentclass[sigconf]{acmart}
\usepackage{balance} 
\usepackage{booktabs} 
\usepackage{multirow}
\usepackage{graphics}
\usepackage{graphicx}
\usepackage{algorithm}
\usepackage{algorithmic}

\usepackage{amsfonts}
\usepackage{color, url}

\usepackage[skip=3pt]{caption}
\usepackage[skip=3pt]{subcaption}

\usepackage{tabto}
\usepackage{url}
\usepackage{xcolor, colortbl}
\usepackage{booktabs}
\usepackage{enumitem}
\usepackage{bbm}
\usepackage{dsfont}

\usepackage{bm}

\usepackage{url}

\usepackage{subcaption}

\usepackage[mathscr]{eucal}

\usepackage{bigstrut,multirow,rotating}
\usepackage{placeins}

\graphicspath{ {figs/} }

\makeatletter
\@namedef{r@tocindent4}{0pt}
\makeatother

\newtheorem{thm}{Theorem}

\DeclareMathOperator*{\argmin}{arg\,min}

\newcommand{\myparatight}[1]{\smallskip\noindent{\bf {#1}:}~}

\copyrightyear{2021}
\acmYear{2021}
\setcopyright{iw3c2w3}
\acmConference[WWW '21]{Proceedings of the Web Conference 2021}{April 19--23, 2021}{Ljubljana, Slovenia}
\acmBooktitle{Proceedings of the Web Conference 2021 (WWW '21), April 19--23, 2021,
	Ljubljana, Slovenia}

\acmPrice{}
\acmDOI{10.1145/3442381.3450066}
\acmISBN{978-1-4503-8312-7/21/04}

\begin{document}

\setlength{\abovedisplayskip}{1.5mm}
\setlength{\belowdisplayskip}{1.5mm}

\title{Data Poisoning Attacks and Defenses to Crowdsourcing Systems
}

\author{Minghong Fang$^{1,2}$, Minghao Sun$^1$, Qi Li$^1$, Neil Zhenqiang Gong$^3$, Jin Tian$^1$, Jia Liu$^2$}
\affiliation{%
	\institution{$^1$Iowa State University \hspace{0.4em}  $^2$The Ohio State University \hspace{0.4em}  $^3$Duke University}
}

\begin{abstract}
A key challenge of big data analytics is how to collect a large volume of (labeled) data. 
Crowdsourcing aims to address this challenge via aggregating and estimating high-quality data (e.g., sentiment label for text) from pervasive clients/users. 
Existing studies on crowdsourcing focus on designing new methods to improve the aggregated data quality from unreliable/noisy clients. 
However, the security aspects of such crowdsourcing systems remain under-explored to date. 
We aim to bridge this gap in this work. 
Specifically, we show that crowdsourcing is vulnerable to \emph{data poisoning attacks}, in which malicious clients provide carefully crafted data to corrupt the aggregated data. 
We formulate our proposed data poisoning attacks as an optimization problem that maximizes the error of the aggregated data.
Our evaluation results on one synthetic and two real-world benchmark datasets demonstrate that the proposed attacks can substantially increase the estimation errors of the aggregated data. 
We also propose two defenses to reduce the impact of malicious clients. 
Our empirical results show that the proposed defenses can substantially reduce the estimation errors of the data poisoning attacks.    	 
\end{abstract}

\ccsdesc[500]{Security and privacy~Systems security}

\keywords{Data poisoning attacks, crowdsourcing, truth discovery
}

\maketitle

\input{introduction}

\input{related}

\input{problem}

\input{attackModel}

\input{experiments}

\input{defense}

\input{conclusion}

\balance
\bibliographystyle{ACM-Reference-Format}
\bibliography{refs}

\end{document}

%% file: introduction.tex

\section{Introduction} \label{sec:intro}
{\bf Background and Motivation:} A well-known challenge for big data analytics is that its success highly relies on a large amount of (labeled) data. Crowdsourcing aims to address this challenge by significantly reducing the labeling cost. However, since individuals may make mistakes, a common practice is to hire multiple workers for the same task and then obtain high-quality aggregated data. Specifically, a data requester (called \emph{central server} in this paper) has a set of items/tasks. The server distributes the items to some \emph{workers}; each worker provides a value for each item that is allocated to him/her; and the server estimates an \emph{aggregated value} for each item using the workers' values. 
For instance, the items could be some text documents and the server aims to estimate a numeric value (e.g., a number between -100 and 100) for each document.

Although crowdsourcing holds a great potential to solve the labeling challenges in big data analytics, a critical challenge that affecting its future large-scale adoption stems from the fact that the data provided by workers in many crowdsourcing applications are {\em noisy} and {\em unreliable}.
To extract high-quality information from unreliable and noisy values provided by workers in crowdsourcing, a widely adopted approach is the so-called truth discovery methods~\cite{dawid1979maximum,dong2009integrating,galland2010corroborating,garcia2017truth,li2014confidence, li2014resolving,li2015truth,li2016survey,li2016conflicts,pasternack2011making,raykar2010learning,wan2016truth,wang2017discovering,yin2008truth,zhao2012probabilistic,zhao2012bayesian}.
Generally speaking, the truth discovery methods cover a family of algorithms that perform {\em weighted} aggregation of worker information based on the quality of each worker.
Specifically, in many real-world crowdsourcing applications, some workers may provide biased or incorrect values for items due to various reasons, e.g., lack of effort, lack of expertise, etc. 
To handle unreliable and noisy worker information, state-of-the-art truth discovery methods (e.g., the conflict resolution on heterogeneous data (CRH)~\cite{li2014resolving,li2016conflicts}, Gaussian truth model (GTM)~\cite{zhao2012probabilistic}, etc.) jointly estimate workers' reliability measured by certain uncertainty metrics (e.g., variance, etc.) and use these reliability metrics as weights in aggregating the values provided by the workers for each item.  
The rationale behind the truth discovery methods is simple: if a worker does not have a large deviation from the majority of the workers very often, then this worker is more likely to be reliable.
Further, if a piece of information is provided by reliable workers, then this information is more likely to be correct and should be assigned a larger weight in the aggregation. 

It is worth noting that, to date, most algorithms in the truth discovery family are based on the assumption that all workers are benign and the unreliable values from the workers are caused by unavoidable randomness in the nature. 
Unfortunately, in the presence of {\em malicious workers} who could provide {\em carefully crafted values} to the server (aka {\em data poisoning attacks}), recent studies have found that existing truth discovery methods could perform rather poorly  (see, e.g., \cite{miao2018attack, miao2018towards}).
We note that, although these results exposed the vulnerability of truth discovery methods,
they remain mostly limited to crowdsourcing applications with categorical labels (i.e., discrete labels in multiple-choice surveys, etc.).
So far, research results on data poisoning attacks and defense for crowdsourcing with {\em continuous} labels are still quite limited.
However, crowdsourcing applications with continuous labeling are prevalent in practice (for example, the temperature values in the crowdsourcing-based weather reporting are continuous). 
In light of the growing importance of crowdsourcing applications, there is a compelling need to investigate and understand the data poisoning attacks and defense for continuous-labeled crowdsourcing systems.

\noindent
{\bf Our work:} 
In this paper, we aim to bridge this gap and show that truth discovery methods are vulnerable to data poisoning attacks, in which an attacker injects malicious workers to a crowdsourcing system and the malicious workers provide carefully crafted values to corrupt the truth discovery system. In particular, our data poisoning attacks can increase the estimation errors of the aggregated values substantially. 
 
Toward this end, we formulate our data poisoning attacks for truth discovery methods as an optimization problem, whose objective function is to maximize the estimation errors of the aggregated values for the attacker-chosen targeted items and whose variables are the values provided by the malicious workers. 
In particular, our optimization problem is bi-level, which is NP-hard and challenging to solve exactly. We address the challenge via iteratively solving the upper-level and lower-level subproblems in our bi-level optimization problem via a projected gradient ascent method. 

We evaluate our data poisoning attacks using one synthetic dataset and two well-known benchmark datasets in the crowdsourcing community. For instance, in one benchmark dataset called Emotion, the central server aims to estimate the sentiment values (ranging from -100 to 100) for 700 documents from 38 workers, where each document is allocated to 10 workers. 
To show the effectiveness of our data poisoning attacks on the truth discovery algorithms, we use two state-of-the-art methods from this family called \emph{Conflict Resolution on Heterogeneous Data (CRH)} and \emph{Gaussian Truth Model (GTM)} as examples. 
We show that our attacks can substantially increase the average estimation error.
%
For instance, on the sentiment estimation dataset, our attack can increase the average estimation error of the sentiment values to 93.69 when 10\% of workers are malicious under the CRH model.   

We also propose two defense mechanisms to mitigate our data poisoning attacks, namely \emph{Median-of-Weighted-Average (MWA)} and \emph{Maximize Influence of Estimation (MIE)}.
In the MWA defense, 
the server first partitions the workers who provide values for a given item into  groups, computes the weighted average in each group, and then estimates the median of the weighted average among the groups as the final aggregated value for the item. 
Note that in MWA, we considers all workers to estimate the aggregated values, though the impact of the malicious workers is mitigated by robust aggregation. 
By contrast, in MIE, we use an influence function to identify the potential malicious workers and remove them before estimating the aggregated values. 
Our empirical results show that our defenses can substantially reduce the estimation errors of our data poisoning attacks. 

Our contributions in this paper are summarized as follows:

\begin{list}{\labelitemi}{\leftmargin=2em \itemindent=-0.3em \itemsep=.2em}
	\item We propose data poisoning attacks to crowdsourcing, which can be formulated as a bi-level optimization problem.
	Due to the NP-hardness of the problem, we propose an efficient algorithm that achieves competitive results. 

	\item We evaluate our attacks on three datasets and show that our attacks can increase the estimation errors substantially. 

	\item  We propose two defenses to mitigate our attacks. Our experimental results demonstrate that the proposed defenses can effectively reduce the estimation errors.
\end{list}

%% file: related.tex

\section{Preliminaries and Related Work} \label{sec:related}
In this section, we first provide an overview and some necessary preliminaries of crowdsourcing and truth discovery methods, using two state-of-the-art truth discovery methods called \emph{Conflict Resolution on Heterogeneous Data (CRH)}~\cite{li2014resolving,li2016conflicts} and \emph{Gaussian Truth Model (GTM)}~\cite{zhao2012probabilistic} as examples. 
Then, we provide an overview on data poisoning attacks, which put our work in comparative perspectives. 
Table~\ref{Main_notations} lists the key notation used in this paper.

\begin{table}[htbp]
	\centering
	\addtolength{\tabcolsep}{-4.0pt}
	\caption{Summary of key notation.}
	\label{Main_notations}%
	\small
	\begin{tabular}{cccc}
		\hline
		Notation &  Definition  \\
		\hline
		$\mathcal{I}$/$\vert \mathcal{I} \vert$   & Set/number of items \\
		$\mathcal{U}$/$\mathcal{\widetilde{U}}$   & Set of normal/malicious workers \\
		$\mathcal{M}$   & Set of all workers, $\mathcal{M} = \mathcal{U} \cup \mathcal{\widetilde{U}}$    \\
		${x_i^u}$/${x}_i^{\tilde{u}}$ & Value of normal/malicious worker $u$/$\tilde{u}$ on item $i$   \\
		$x^{\star}_i$/$\widehat{x}^{*}_i$ & Aggregated value for item $i$ before/after attack \\
		$w_u$/${w}_{\tilde{u}}$ & Weight of normal/malicious worker $u$/$\tilde{u}$ \\
		${\sigma}_u^2$/${\sigma}_{\tilde{u}}^2$  & Variance of normal/malicious worker $u$/$\tilde{u}$ \\
		$\mathcal{U}_i$/$\mathcal{\widetilde{U}}_i$   & Set of normal/malicious workers who provide values for item $i$    \\
		$\mathcal{I}_u$/$\mathcal{{I}}_{\tilde{u}}$   & Set of items rated by normal/malicious worker $u$/$\tilde{u}$    \\
		\hline
	\end{tabular}%
\end{table}%

\subsection{The Truth Discovery Methods: A Primer}

In this subsection, we provide an overview on the family of truth discovery methods for crowdsourcing.
In most crowdsourcing systems, there is a {\em central server} performing data aggregation and there are some clients called \emph{workers}. 
We denote the set of workers as $\mathcal{U}$. The server has a set of items $\mathcal{I}$ and aims to estimate a certain \emph{value} for each item based on the input from the workers. 
In this work, we focus on the cases where the value to be estimated is {\em continuous}. 
For instance, the items could be a set of text documents and the server aims to estimate a numeric value for each document from the workers.  
The server assigns each item to a subset of workers. We denote by $\mathcal{I}_u$ the set of items that are allocated to worker $u$. Moreover, we denote by ${x_i^u}$ the value that the worker $u$ provides for item $i$, where  $u \in \mathcal{U}$ and $i \in \mathcal{I}_u$.

To find reliable information among unreliable data, a naive approach is majority voting or taking the average of the values provided by workers. 
A major limitation of these methods is that they do not take the quality/reliability of workers into consideration. In practice, the quality of different workers varies. 
To address this challenge, the truth discovery approaches are proposed to automatically jointly estimate the quality of workers while performing information aggregation.
The rationale behind the truth discover methods is to characterize the reliability of a worker as a weight. 
If a worker has a smaller weight, then all of its provided values are less reliable. 
To illustrate the basic idea of the truth discovery methods, in what follows, we use two state-of-the-art algorithms in this family that are widely adopted in crowdsourcing systems as concrete examples: i) conflict resolution on heterogeneous data (CRH)~\cite{li2014resolving,li2016conflicts} and ii) Gaussian truth model (GTM)~\cite{zhao2012probabilistic}.

\subsubsection{The Conflict Resolution on Heterogeneous Data Model (CRH)} 
CRH, a state-of-the-art truth discovery method, jointly estimates the aggregated values of items and the weights of workers. In particular, CRH  formulates the estimations of the aggregated values and worker weights as the following optimization problem:
\begin{equation}
\begin{split}
\min_{X^{\star}, W} & f(X^{\star},W)=\sum\nolimits_{u \in \mathcal{U}} w_u \sum\nolimits_{i\in {\mathcal{I}_u}}d(x_i^u, x_i^\star) \\
&  \text{s.t.} \quad \sum\nolimits_{u \in \mathcal{U}} \exp(-w_u)=1,
\label{CRH_orgin}
\end{split}
\end{equation}
where $X^{\star} = \left\{ {x_i^{\star}} \right\}_{i \in \mathcal{I}}$ is the set of {aggregated values} for all the items, $W = \left\{ {w_u} \right\}_{u \in \mathcal{U}}$ is the set of weights for all the workers, $w_u$ is weight for worker $u$, $d(\cdot)$ is a function to measure the distance between a worker's value of an item and the item's aggregated value, which reflects the reliability of this particular worker. 
In our experiments, we use the square distance function. CRH solves the optimization problem by iteratively alternating between the following two steps:

\myparatight{Step 1 (Estimate the aggregated values)} In this step, the workers' weights $W$ are fixed, and the aggregated value for item $i$ is updated as follows:
\begin{align}
x_i^{\star}=\frac{{\sum\nolimits_{u \in {\mathcal{U}_i}} {{w_u}x_i^u} }}{{\sum\nolimits_{u \in {\mathcal{U}_i}} {w_u} }},
\label{CRH_orgin_update_truth}
\end{align}
where $\mathcal{U}_i$ is the set of workers who provide values for item $i$.

\myparatight{ Step 2 (Update worker weights)} Next, the aggregated values $X^{\star}$ are fixed, and the weight of worker $u$ is updated as follows:
\begin{align} 
{w_u} = \log \left( {\frac{{\sum\nolimits_{k \in \mathcal{U}} {\sum\nolimits_{i\in {\mathcal{I}_k}} {d(x_i^k,x_i^{\star})} } }}{{\sum\nolimits_{i\in {\mathcal{I}_u}} {d(x_i^u,x_i^{\star})} }}} \right).
\label{CRH_orgin_update_weight}
\end{align}


It can be seen from \eqref{CRH_orgin_update_weight} that, the smaller the distance $d(x_i^u,x_i^{\star})$, the larger the weight of worker $u$ (i.e., more reliable).
We can use the block coordinate descent method~\cite{bertsekas1997nonlinear} to iteratively update the above two-step procedure until some convergence criterion is met. Algorithm~\ref{CRH_framework} shows the CRH framework. 
In this paper, we assume that all workers are given equal initial weights in the CRH method.

\begin{algorithm}[t!]
	\caption{The CRH framework \protect\cite{li2014resolving,li2016conflicts}.}
	\label{CRH_framework}
	\begin{algorithmic}[1]
		\renewcommand{\algorithmicrequire}{\textbf{Input:}}
		\renewcommand{\algorithmicensure}{\textbf{Output:}}
		\REQUIRE Values from workers ${x_i^u}$ for $u \in \mathcal{U}, i \in \mathcal{I}_u$.
		\ENSURE  Aggregated values $X^{\star}$ and worker weights $W$.
		\STATE  Server initializes the worker weights.
		\WHILE {the convergence condition is not satisfied}
		\STATE  Server updates the aggregated value of each item according to Eq.~(\ref{CRH_orgin_update_truth}).
		\STATE Server updates the weight of each worker according to Eq.~(\ref{CRH_orgin_update_weight}).
		\ENDWHILE
		\RETURN $X^{\star}$ and $W$. 
	\end{algorithmic} 
\end{algorithm}

\subsubsection{The Gaussian Truth Model (GTM)} 

GTM model is a Bayesian probabilistic model designed for numeric data in truth discovery. 
In GTM, the reliability of a worker are captured by a variance parameter. Intuitively, a worker with larger variance is more likely to provide inaccurate values that deviate more from the truth. 
The GTM model first normalizes all input values to its z-scores, then tries to solve the following optimization problem:
\begin{align}
&\min_{X^{\star}, \Omega}  f(X^{\star},\Omega) \propto - \sum\nolimits_{u \in {\mathcal{U}}} \left( 2(\alpha + 1)\log{{\sigma}_u} + \frac{\beta}{{\sigma}_u^2} \right) - \nonumber \\ 
& \sum\nolimits_{i \in \mathcal{I}} \! \frac{( x_i^{\star} - \mu_0 )^2}{2{\sigma}_0^2}
-
\sum\nolimits_{i \in {\mathcal{I}}} \! \sum\nolimits_{u \in {\mathcal{U}_i}}\left( \log{{\sigma}_u} + \frac{( x_i^u - x_i^{\star} )^2}{2{\sigma}_u^2} \right),
\label{GTM_orgin}
\end{align}
where $\Omega = \left\{ {{\sigma}_u^2} \right\}_{u \in \mathcal{U}}$ is the set of variances for all the workers, ${\sigma}_u^2$ is the variance of worker $u$,
$\mu_0$ and ${\sigma}_0^2$ are prior parameters and $\alpha$ and $\beta$ are hyper-parameters.
The GTM model leverages the EM algorithm~\cite{dempster1977maximum} which contains the following expectation step (E-step) and maximization step (M-step) to  iteratively update aggregated values and variance parameters of workers':
 
\myparatight{E-step (Estimate the aggregated values)}  In this step, the workers' variances are fixed, and the aggregated value for item $i$ is computed by solving $\frac{\partial f}{x_i^{\star}}=0$, which yields:
\begin{align}
x_i^{\star}=\frac{\frac{\mu_0}{{\sigma}_0^2} +  \sum\nolimits_{u \in {\mathcal{U}_i}}  \frac{x_i^u}{{\sigma}_u^2}  }  
{\frac{1}{{\sigma}_0^2} +  \sum\nolimits_{u \in {\mathcal{U}_i}}  \frac{1}{{\sigma}_u^2}}.
\label{GTM_orgin_update_truth}
\end{align}

\myparatight{M-step (Update worker variances)}  In this step, the aggregated values $X^{\star}$ are fixed, and the variance of worker $u$ can be calculated by solving $\frac{\partial f}{{\sigma}_u^2}=0$, which yields:
\begin{align}
{\sigma}_u^2 = \frac{2\beta + \sum\nolimits_{i \in \mathcal{I}_u}(x_i^u - x_i^{\star})^2}{2(\alpha+1) + \vert \mathcal{I}_u \vert},
\label{GTM_orgin_update_variance}
\end{align}
where $\vert \mathcal{I}_u \vert$ is the number of values provided by worker $u$.
The EM algorithm alternates between the above two steps iteratively
until some convergence criterion is satisfied. 
The GTM algorithmic framework is similar to Algorithm~\ref{CRH_framework}, and we omit it here for briefly.

\subsection{Data Poisoning Attacks: An Overview} \label{sec:data_poisoning}

Generally speaking, data poisoning attacks refer to manipulating data to corrupt certain computational results based on those data. 
For instance, in machine learning, a classifier is learnt using a training dataset; and a data poisoning attack can carefully forge the training dataset to corrupt the learnt classifier~\cite{alfeld2016data,biggio2012poisoning,chen2017targeted,jagielski2018manipulating,shafahi2018poison,xiao2015feature}. 
In data poisoning attacks to recommender systems~\cite{christakopoulou2019adversarial,li2016data, yang2017fake,fang2018poisoning,fang2020influence}, an attacker can inject fake users with carefully crafted rating values to a recommender system such that the recommender system makes recommendations as the attacker desires, e.g., recommending an attacker-chosen item to many normal users. In federated learning, an attacker can inject malicious workers with misleading training samples to corrupt the learnt global model~\cite{baruch2019little,bhagoji2018analyzing,fang2019local,xie2019fall,cao2020fltrust}. 
Different computations (e.g., machine learning, recommender system, and federated learning) often require different data poisoning attacks to optimize the attack effectiveness. 
The most relevant to ours are~\cite{miao2018attack, miao2018towards,tahmasebian2020crowdsourcing}, where the authors proposed efficient attack algorithms that reduce the effectiveness of crowdsourcing systems with strategic malicious workers. However, the proposed attack models focus on categorical data and do not consider the potential defense deployed by the server.

%% file: problem.tex

\section{Problem Formulation} \label{sec:problem}

In this section, we first introduce our threat model, including the attacker's goal, capability, and knowledge.
Then, we formulate the attacker's goal mathematically.

\subsection{Threat Model}

\begin{figure}[!t]
	\centering
	\includegraphics[scale = 0.28]{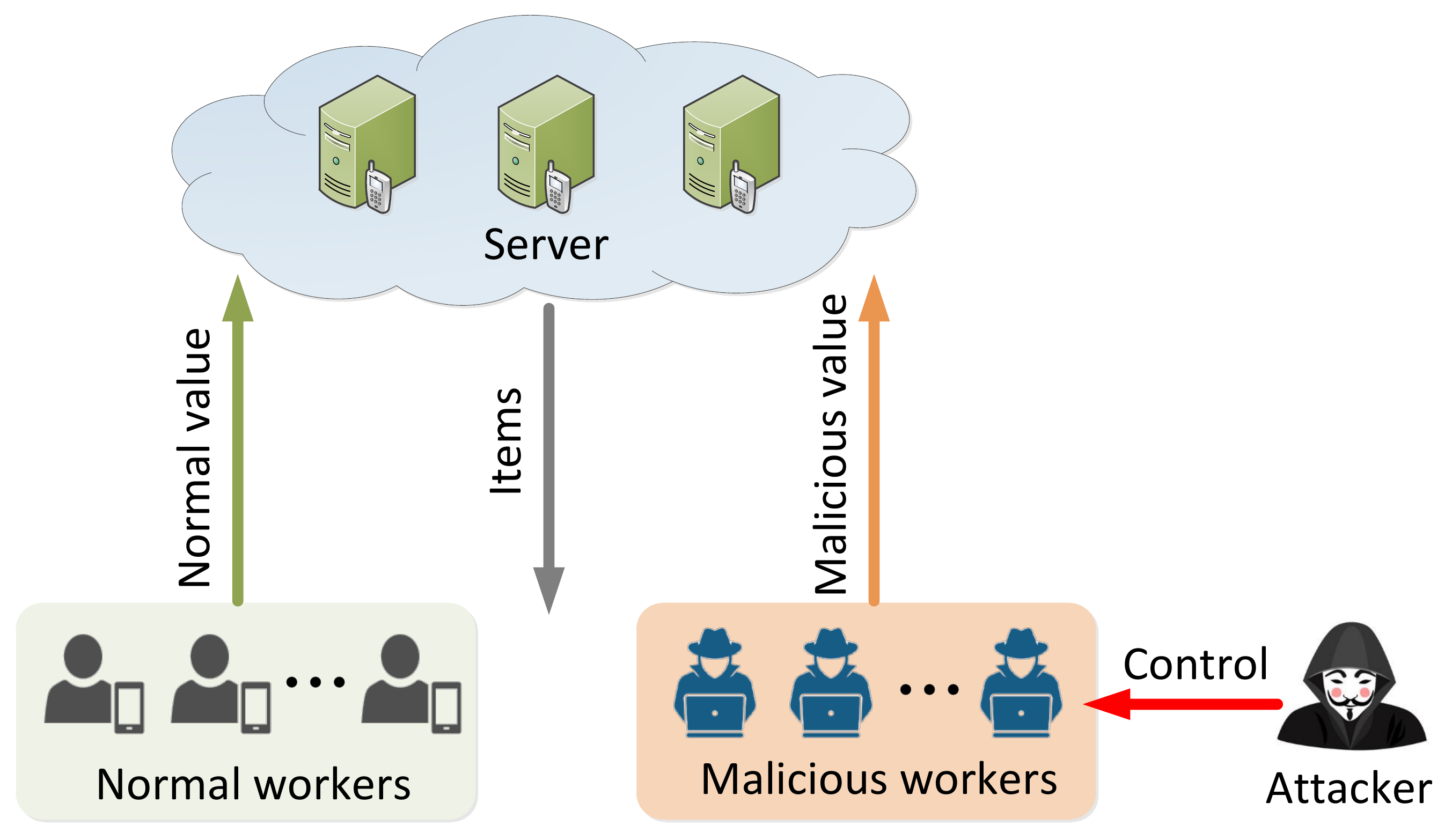}
	\caption{Crowdsourcing systems under attack.}\label{Crowd_sensing_under_attack}
\end{figure}

\myparatight{Attacker's goal} Given a targeted crowdsourcing system, the goal of the attacker is to maximize the estimation errors of the aggregated values for some attacker-chosen targeted items. 
This attack is well motivated in real-world crowdsourcing systems. For instance, an attacker may be interested in manipulating the real-time navigation crowdsourcing system such that the system  provides misleading or even life-threatening directions to users. An attacker could also attack a competitor's system to gain competitive advantages. 

\myparatight{Attacker's capability and knowledge}
In our threat model, we assume the attacker is able to inject some malicious workers into the crowdsourcing system and launch attacks by carefully crafting their values, as shown in Figure~\ref{Crowd_sensing_under_attack}. This threat model is practical because a crowdsourcing system is essentially a distributed system and an attacker can inject malicious workers into it, 
which is more realistic than modifying existing data of normal workers.

The attacker can have different degrees of knowledge of the targeted crowdsourcing system. In particular, we consider two cases, \emph{full knowledge} and \emph{partial knowledge}. 
In the full knowledge scenario, the attacker has full knowledge of the aggregation method and all values provided by normal workers.
We note that, although appearing to be a strong assumption, the full-knowledge setting is \emph{not} uncommon in practice since all data and the sources of the data could be public in crowdsourcing.
For instance, in the weather forecast integration task, the workers in different monitoring stations collect the local weather data and upload their data to the crowdsourcing platform (e.g., website). These data are available to all workers and each worker can see weather information at other locations and know where the data comes from.
%
%

In the partial knowledge scenario, the attacker knows the aggregation method but only knows the values of a subset of normal workers. For instance, the attacker may compromise some normal workers via bribing them, compromising their computer systems, and/or stealing their credentials.

\subsection{Formulating Data Poisoning Attacks}
We  formulate our data poisoning attacks as an optimization problem, which maximizes the estimation errors of the targeted items' aggregated values. 
Suppose that $\mathcal{T}$ is the set of attacker-chosen targeted items and $\alpha$ fraction of workers are malicious. 
Specifically, we let $\mathcal{U}$ and $\mathcal{\widetilde{U}}$ denote the sets of normal and malicious workers, respectively. 
Then, we have  $\vert \mathcal{\widetilde{U}} \vert = \left  \lfloor \frac{{\alpha \vert \mathcal{U} \vert }}{{1 - \alpha }} \right \rfloor $. 
We let ${x}_t^{\tilde{u}} $ denote the value that a malicious worker $\tilde{u} \in \mathcal{\widetilde{U}}$ provides on item $t \in \mathcal{T}$.
Then, the attacker's goal is to find an optimal value for each malicious worker to rate each targeted item, such that after injecting those malicious workers into the crowdsourcing system, the distance of the aggregated values before and after our attack is maximized.
We let $x^{\star}_t$ and $\widehat{x}^{*}_t$ denote the aggregated values for item $t$ before and after the poisoning attack, respectively.
Then, we can formulate our attack as follows:
\begin{align}
\label{attack_goal}
&  \underset{\left\{ {x}_t^{\tilde{u}} \right\}_{t \in \mathcal{T},\tilde{u} \in \mathcal{\widetilde{U}}} } {\text{Maximize}}  \frac{1}{\left| \mathcal{T} \right|} \sum\limits_{t \in {\mathcal{T}}} d(\widehat{x}^{*}_t, x^{\star}_t )     \\
&\text{s.t. }    \vert \mathcal{\widetilde{U}} \vert = \left \lfloor \frac{{\alpha \vert \mathcal{U} \vert }}{{1 - \alpha }} \right \rfloor ,
\end{align}
where $d(\cdot)$ is a distance function that measures the estimation error of an item introduced by our attack. In our experiments, we use the square distance function. The objective function $\frac{1}{\left| \mathcal{T} \right|} \sum\nolimits_{t \in {\mathcal{T}}} d(\widehat{x}^{*}_t, x^{\star}_t )$ measures the average estimation error for the targeted items introduced by the attack. Note that our problem formulation in Eq.~(\ref{attack_goal}) can be applied to any truth discovery method. Solving the optimization problem for a given truth discovery method (e.g., CRH, GTM, etc.) leads to a specific data poisoning attack to this particular method.

%% file: attackModel.tex

\section{Our Attacks} \label{attackModel}

In this section, we introduce our data poisoning attacks in the full-knowledge and partial-knowledge scenarios, respectively. Note that our attacks are essentially to solve the optimization problem in Eq.~(\ref{attack_goal}) under these two settings.

\subsection{Full-Knowledge Attack}

In the full-knowledge scenario, the attacker knows the values provided by the normal workers. 
Note that the attacker's goal in Eq.~(\ref{attack_goal}) is to maximize the deviation of the truth discovery's output before and after the attack, while the server's goal of Eq.~(\ref{CRH_orgin}) or Eq.~(\ref{GTM_orgin}) is to estimate the aggregated value for each item.
Moreover, the malicious workers' values should be within the normal range to avoid outlier detection. 
Considering these two goals together, we can instantiate the optimization problem in Eq.~(\ref{attack_goal}) for the CRH and GTM models as the following bi-level optimization problem for the attacker:
\begin{align} 
\label{bi_level_opt}
& \underset{\left\{ {x}_t^{\tilde{u}} \right\}_{t \in \mathcal{T},\tilde{u} \in \mathcal{\widetilde{U}}} } {\text{Maximize}}  \frac{1}{\left| \mathcal{T} \right|} \sum\limits_{t \in {\mathcal{T}}} d(\widehat{x}^{*}_t, x^{\star}_t )     \\
&  \text{s.t.} \! \!\quad {x}_t^{\tilde{u}} \in \left[ x_t^{\min}, x_t^{\max} \right], \forall \tilde{u} \in \mathcal{\widetilde{U}}, \forall t \in \mathcal{T}, \nonumber \\
& \quad\quad \{ \widehat{X}^{*}, \Lambda^* \} = 
\argmin\limits_{\widehat{X}^{\star}, \Lambda}  f(\widehat{X}^{\star}, \Lambda), \nonumber
\end{align}
where $x_t^{\min}$ and $x_t^{\max}$ are the minimum and maximum values of item $t$ provided by normal workers, respectively; $\Lambda^*$ are the after-attack weight or variance parameters for all workers. 
We remind here that Problem~(\ref{bi_level_opt}) is a general attack framework, and can be applied to any truth discovery method.
To be specific, for the CRH model, we just substitute $f(\widehat{X}^{\star}, W)$ in Eq.~(\ref{CRH_orgin}) to $f(\widehat{X}^{\star}, \Lambda)$ in Eq.~(\ref{bi_level_opt}) and let $W=\Lambda$; for the GTM model, we substitute $f(\widehat{X}^{\star}, \Omega)$ in Eq.~(\ref{GTM_orgin}) to $f(\widehat{X}^{\star}, \Lambda)$ in Eq.~(\ref{bi_level_opt}) and let $\Omega=\Lambda$.
From Eq.~(\ref{bi_level_opt}), we can see that the upper-level problem is to determine the optimal fake values for malicious workers, and the lower-level problem is to estimate the aggregated value for each item.
The lower-level problem can be considered a constraint of the upper-level problem. 
Bi-level optimization is NP-hard in general~\cite{hansen1992new}. 
In our bi-level formulation, although the upper-level problem is relatively simple, the lower-level problem is highly non-linear and non-convex.
In this paper, we propose the following two-step iterative method to solve the above bi-level optimization problem:

\myparatight{Step 1 (Update the aggregated values and worker weights/var-iances)} In this step, the attacker fixes the malicious workers' values, then solves the lower-level optimization problem to obtain aggregated values and worker weights/variances $\{ \widehat{X}^{*}, \Lambda^*\}$. As we have discussed above, this step can be done for the CRH model by solving Eqs.~(\ref{CRH_orgin_update_truth})-(\ref{CRH_orgin_update_weight}) or for the GTM model by solving Eqs.~(\ref{GTM_orgin_update_truth})-(\ref{GTM_orgin_update_variance}).

\myparatight{Step 2 (Update malicious workers' values)} In this step, we let
$\widetilde{X} =  \left\{ {x_t^{\tilde{u}}} \right\}_{t \in \mathcal{T}, \tilde{u} \in \mathcal{\widetilde{U}}}$ denote the values of all malicious workers, and define the following objective function:
\begin{align}
\mathcal{L}(\widetilde{X}) = \sum\nolimits_{t \in {\mathcal{T}}} d(\widehat{x}^{*}_t, x^{\star}_t ).
\label{loss_function_L}
\end{align}

When taking the malicious workers into consideration, we can compute the term $\widehat{x}^{*}_t$ in Eq.~(\ref{loss_function_L}) as followings for the CRH model:
\begin{align}
\widehat{x}_i^{*}=
\frac{{\sum\nolimits_{u \in {\mathcal{U}_i}} {{w_u}x_i^u}  } + {\sum\nolimits_{\tilde{u} \in \mathcal{\widetilde{U}}_i} {{{w}_{\tilde{u}}}{x}_i^{\tilde{u}}} }}
{{\sum\nolimits_{u \in {\mathcal{U}_i}} {w_u} } + {\sum\nolimits_{\tilde{u} \in \mathcal{\widetilde{U}}_i} {{w}_{\tilde{u}}} }}.
\label{CRH_att_update_truth}
\end{align}

For the GTM model, $\widehat{x}^{*}_t$ can be computed as:
\begin{align}
\widehat{x}_i^{*}=\frac{\frac{\mu_0}{{\sigma}_0^2} +  \sum\nolimits_{u \in {\mathcal{U}_i}}  \frac{x_i^u}{{\sigma}_u^2} + \sum\nolimits_{\tilde{u} \in {\mathcal{\widetilde{U}}_i}}  \frac{x_i^{\tilde{u}}}{{\sigma}_{\tilde{u}}^2} }  
{\frac{1}{{\sigma}_0^2} +  \sum\nolimits_{u \in {\mathcal{U}_i}}  \frac{1}{{\sigma}_u^2} + \sum\nolimits_{\tilde{u} \in {\mathcal{\widetilde{U}}_i}}  \frac{1}{{\sigma}_{\tilde{u}}^2} }.
\label{GTM_att_update_truth}
\end{align}

For malicious worker $v \in \widetilde{U}$, the gradient of the objective function $\mathcal{L}(\widetilde{X})$ with respect to ${x}_t^v$ can be computed as follows:
\begin{align}
\!\! {\nabla _{{x}_t^v}}\mathcal{L} (\widetilde{X})
\!=\! \sum\limits_{t^{\prime} \in {\mathcal{T}}} 
\frac{{\partial d(\widehat{x}^{*}_{t^{\prime}}, x^{\star}_{t^{\prime}} )}} {{\partial {{x}_t^v}}}
\!=\! \sum\limits_{t^{\prime} \in {\mathcal{T}}} 
\frac{{\partial d(\widehat{x}^{*}_{t^{\prime}}, x^{\star}_{t^{\prime}} )}}{{\partial {\widehat{x}^{*}_{t^{\prime}}}}} 
\cdot 
\frac{{\partial {\widehat{x}^{*}_{t^{\prime}}}}} {{\partial {{x}_t^v}}}.
\label{step2_upper_level}
\end{align}

Here, if we adopt the square distance function $d(\widehat{x}^{*}_{t^{\prime}}, x^{\star}_{t^{\prime}} )={(\widehat{x}^{*}_{t^{\prime}} - x^{\star}_{t^{\prime}})^2}$, then we have:
\begin{align}
\frac{{\partial d(\widehat{x}^{*}_{t^{\prime}}, x^{\star}_{t^{\prime}} )}}{{\partial {\widehat{x}^{*}_{t^{\prime}}}}}  = 2(\widehat{x}^{*}_{t^{\prime}} - x^{\star}_{t^{\prime}} ).
\end{align}

For the CRH model, from Eq.~(\ref{CRH_att_update_truth}), the gradient $\frac{{\partial {\widehat{x}^{*}_{t^{\prime}}}}} {{\partial {{x}_t^v}}} $ can be calculated as:
\begin{equation}
\frac{{\partial {\widehat{x}^{*}_{t^{\prime}}}}} {{\partial {{x}_t^v}}} =
\left\{
\begin{matrix}
\frac{ {w}_v }
{{\sum\nolimits_{u \in {\mathcal{U}_t}} {w_u} } + {\sum\nolimits_{\tilde{u} \in \mathcal{\widetilde{U}}_t} {{w}_{\tilde{u}}} }}, &  t^{\prime}=t,\\
0, &  \text{otherwise}.
\end{matrix}
\right.
\end{equation}

For the GTM model, according to Eq.~(\ref{GTM_att_update_truth}), the gradient $\frac{{\partial {\widehat{x}^{*}_{t^{\prime}}}}} {{\partial {{x}_t^v}}} $ can be calculated as:

\begin{equation}
\frac{{\partial {\widehat{x}^{*}_{t^{\prime}}}}} {{\partial {{x}_t^v}}} =
\left\{
\begin{matrix}
\frac{ 1}  
{ {\sigma}_v^2 \left( \frac{1}{{\sigma}_0^2} +  \sum\limits_{u \in {\mathcal{U}_t}}  \frac{1}{{\sigma}_u^2} + \sum\limits_{\tilde{u} \in {\mathcal{\widetilde{U}}_t}}  \frac{1}{{\sigma}_{\tilde{u}}^2} \right)}, &  t^{\prime}=t,\\
0, &  \text{otherwise}.
\end{matrix}
\right.
\end{equation}

After obtaining ${\nabla _{{x}_t^v}}\mathcal{L}(\widetilde{X})$, we then can use the projected gradient ascent method to update the value of malicious worker $v$ as follows:
\begin{align}
{x}_t^v [r+1] = \text{Proj}_{\left[ x_t^{\min}, x_t^{\max} \right]} \left(  {x}_t^v [r] + {\eta _r} \cdot {\nabla _{{x}_t^v}}\mathcal{L} (\widetilde{X}) \right) ,
\label{gradient_ascent}
\end{align}
where $r$ is the $r$-th iteration, $\text{Proj}_{\left[ x_t^{\min}, x_t^{\max} \right]}(\cdot) $ is the projection operator onto the range $\left[ x_t^{\min}, x_t^{\max} \right]$, and $\eta _r$ is the step size in the $r$-th iteration. 
In this paper, we assume that the attacker initializes the workers with equal weights for the CRH model and equal variances for the GTM model.
Algorithm~\ref{full_know_poisoning_attack_alg} summaries our full-knowledge attack algorithm for the CRH model.
The full-knowledge attack algorithm for the GTM model is similar to Algorithm~\ref{full_know_poisoning_attack_alg}, and we omit it here to avoid repetitiveness.

\begin{algorithm}[t!]
	\caption{{Full-knowledge data poisoning attack.}}
	\label{full_know_poisoning_attack_alg}
	\begin{algorithmic}[1]
		\renewcommand{\algorithmicrequire}{\textbf{Input:}}
		\renewcommand{\algorithmicensure}{\textbf{Output:}}
		\REQUIRE Values from all normal workers ${x_i^u}$ for $u \in \mathcal{U}, i \in \mathcal{I}_u$.
		\ENSURE  Malicious workers' values.
		\STATE The attacker initializes the workers' weights.
		\STATE The attacker estimates the aggregated values before attack by iteratively solving Eqs.~(\ref{CRH_orgin_update_truth})-(\ref{CRH_orgin_update_weight}). 
		\WHILE {the convergence condition is not satisfied}
		\STATE The attacker computes the aggregated values and workers' weights $\{ \widehat{X}^{*}, W^*\}$ by iteratively solving Eqs.~(\ref{CRH_orgin_update_truth})-(\ref{CRH_orgin_update_weight}).
		\STATE The attacker estimates the gradient ${\nabla _{{x}_t^v}}\mathcal{L}(\widetilde{X})$ according to Eq.~(\ref{step2_upper_level}).
		\STATE The attacker updates malicious workers' values according to Eq.~(\ref{gradient_ascent}).
		\ENDWHILE
		\RETURN ${x}_t^v$ for $v \in \widetilde{U}, t \in \mathcal{T}$. 
	\end{algorithmic} 
\end{algorithm}

\subsection{Partial-Knowledge Attack}
In the previous section, we showed that the attacker can launch efficient data poisoning attacks to CRH and GTM truth discovery methods when the attacker has full knowledge of the targeted system.
However, this could be a restrictive assumption in practice.
In this section, we consider a weaker assumption that the attacker only observes part of the values of normal workers on the targeted items. 
Further, we note that in the partial-knowledge attack, the attacker only needs access to the values of normal workers for the targeted items, i.e., the attacker does not need to know the values of normal workers for non-targeted items.

\subsubsection{Aggregated Values Estimation with Bootstrapping}

In the partial-knowledge attack, for a given targeted item, it is hard for the attacker to estimate the aggregated value accurately since he only has access to part of the values provided by normal workers.
To address this challenge, we leverage the Bootstrapping \cite{cheng2014parallel,efron1992bootstrap,efron1994introduction,hogg2005introduction} technique to obtain more accurate before-attack estimated aggregated values for targeted items.
Bootstrapping is a classic re-sampling method to estimate a sample distribution. 
The basic idea of Bootstrapping is to independently sample with replacement from an observed dataset with the same sample size, and perform estimation among these resampled data.

We let $\mathcal{S}_t \in \mathcal{U}_t$ denote the set of normal workers whose
observations on the targeted item $t \in \mathcal{T}$ can be accessed by the attacker.
In our attack model, once the normal workers' weights are held fixed, the attacker uses the Bootstrapping method to obtain $B$ estimated aggregated value for item $t$. 
To be specific, the{\tiny } $b$-th estimation $\widehat{x}^{b}_t$, $1 \le b \le B$, can be calculated by the following two steps:

\myparatight{Step 1 (Workers Bootstrapping)} The attacker randomly samples a set of normal workers $\mathcal{S}_t^b$ from $\mathcal{S}_t$ with replacement, such that $\vert \mathcal{S}_t^b \vert = \vert \mathcal{S}_t \vert$.

\myparatight{Step 2 (Value Estimation)} The attacker computes the aggregated value for item $t$ in the $b$-th estimation $\widehat{x}^{b}_t$ according to Eq.~(\ref{CRH_orgin_update_truth}) for the CRH model or Eq.~(\ref{GTM_orgin_update_truth}) for the GTM model based on the sampled values provided by workers in set $\mathcal{S}_t^b$.

After the attacker repeats the above two-step procedure $B$ times, the attacker obtains $B$ estimated aggregated values for item $t$. 
Then, the final before-attack estimated aggregated value $\widehat{x}_t^{boot}$ can be computed as:
\begin{align}
\widehat{x}_t^{boot} = \frac{1}{B}\sum\nolimits_{b=1}^B \widehat{x}^{b}_t. 
\label{x_boot}
\end{align}

After the attacker uses the Bootstrapping technique to estimate the before-attack value of each item, the attacker then uses the projected gradient ascent method to update malicious workers' value.  
Algorithm \ref{partial_know_poisoning_attack_alg} summaries our partial-knowledge attack algorithm for the CRH model, and the partial-knowledge attack algorithm for the GTM model follows a similar procedure.
Note that we do not leverage the Bootstrapping method to estimate the after-attack value, since the majority of value of normal workers may be drawn from a certain distribution, while value of malicious workers do not necessarily fit the distribution.

\begin{algorithm}[t!]
	\caption{{Partial-knowledge data poisoning attack.}}
	\label{partial_know_poisoning_attack_alg}
	\begin{algorithmic}[1]
		\renewcommand{\algorithmicrequire}{\textbf{Input:}}
		\renewcommand{\algorithmicensure}{\textbf{Output:}}
		\REQUIRE Part of values on the targeted items provided by normal workers in set $\mathcal{S}_t$, $ t \in \mathcal{T}$, $B$.
		\ENSURE  Malicious workers' values.
		\STATE The attacker initializes the workers' weights.
		
		 //Estimate the aggregated values before attack.
		\WHILE {the convergence condition is not satisfied}
		\FOR {each $t \in \mathcal{T}$} 
		\FOR {$b=1,\cdots,B$} 
		\STATE The attacker first bootstraps $\mathcal{S}_t^b$ from $\mathcal{S}_t$, then computes $\widehat{x}^{b}_t$ according to Eq.~(\ref{CRH_orgin_update_truth}).
		\ENDFOR
		\STATE The attacker computes $\widehat{x}_t^{boot}$ according to Eq.~(\ref{x_boot}).
		\ENDFOR
		\STATE The attacker updates the weight of each normal worker according to Eq.~(\ref{CRH_orgin_update_weight}).
		\ENDWHILE
		
		 //Update malicious workers' values.
		\WHILE {the convergence condition is not satisfied}
		\STATE The attacker computes the aggregated values and workers' weights $\{ \widehat{X}^{*}, W^*\}$ by iteratively solving Eqs.~(\ref{CRH_orgin_update_truth})-(\ref{CRH_orgin_update_weight}).
		\STATE The attacker estimates the gradient ${\nabla _{{x}_t^v}}\mathcal{L}(\widetilde{X})$ according to Eq.~(\ref{step2_upper_level}).
		\STATE The attacker updates malicious workers' values according to Eq.~(\ref{gradient_ascent}).
		\ENDWHILE
		\RETURN ${x}_t^v$ for $v \in \widetilde{U}, t \in \mathcal{T}$. 
	\end{algorithmic} 
\end{algorithm}

\subsubsection{Convergence in Distribution}

In this section, we show that for the CRH model, the aggregated value estimated by the Bootstrapping technique converges in distribution to the aggregated value computed by all the observed values at once.
We discuss it for the targeted item $t$, and it can be applied to other targeted items.
Note that in this section, we do not assume that the values provided by workers are independent and identically distributed.

We assume that the value of normal worker $u$ on item $t$ follows a normal distribution, i.e., $x_t^u \sim N\left( {x}^{\star}_t, {\sigma}_u^2 \right)$, where the variance ${\sigma}_u^2$ measures the quality of values provided by worker $u$, $u \in \mathcal{S}_t$.
We let $X_{\mathcal{S}_t}$ denote the values provided by workers in set $\mathcal{S}_t$, and let $\widehat{\theta} \left( X_{\mathcal{S}_t}  \right)  $ denote the estimator that the attacker uses.
From Eq.~(\ref{CRH_orgin_update_truth}), the estimator $\widehat{\theta} \left( X_{\mathcal{S}_t}  \right)$ of item $t$ can be computed as
$
\widehat{\theta} \left( X_{\mathcal{S}_t}  \right)=\frac{{\sum\nolimits_{u \in {\mathcal{S}_t}} {{w_u}x_t^u} }}{{\sum\nolimits_{u \in {\mathcal{S}_t}} {w_u} }}.
\label{CRH_orgin_update_truth_theta}
$
Since $x_t^u \sim N\left( {x}^{\star}_t, {\sigma}_u^2 \right)$, we have
$
\mathbb{E} \left[  \widehat{\theta} \left( X_{\mathcal{S}_t}  \right) \right] = x_t^{\star}
,$
$
\text{Var}\left(  \widehat{\theta} \left( X_{\mathcal{S}_t}  \right) \right) =\frac{{\sum\nolimits_{u \in {\mathcal{S}_t}} {{w_u^2} {\sigma}_u^2} }}  { \left( \sum\nolimits_{u \in {\mathcal{S}_t}} {w_u}  \right) ^2}.
$
We further define 
$
\widehat{ \text{Var}}\left(  \widehat{\theta} \left( X_{\mathcal{S}_t}  \right) \right)
\stackrel{\text{def}} = 
\frac{{\sum\nolimits_{u \in {\mathcal{S}_t}} {{w_u^2} \widehat{\sigma}_u^2} }}  { \left( \sum\nolimits_{u \in {\mathcal{S}_t}} {w_u}  \right) ^2},
$
where $\widehat{\sigma}_u^2 = \frac{{\sum\nolimits_{t \in \mathcal{I}_u} \left(   x_t^u -\widehat{x}_t^{boot} \right)^2 }}{\vert \mathcal{I}_u \vert - 1}$ and $\widehat{x}_t^{boot}$ can be computed by Eq.~(\ref{x_boot}). 
To measure the error between $\widehat{\theta} \left( X_{\mathcal{S}_t}  \right)$ and ${x}^{\star}_t$, we construct a statistic $Q$ as follows:
\begin{align}
Q = \frac{\widehat{\theta} \left( X_{\mathcal{S}_t}  \right) - x_t^{\star}}
{ \left[  \widehat{ \text{Var}}\left (  \widehat{\theta} \left( X_{\mathcal{S}_t}  \right) \right)  \right]   ^{1/2} /\sqrt{\vert \mathcal{S}_t \vert}}.
\label{Q_distribution}
\end{align}

Since the distribution of $Q$ is usually unknown \emph{a priori}, the attacker could leverage the Bootstrapping strategy to approximate $Q$. 
Note that in the $b$-th estimation of $\widehat{x}^{b}_t$, the attacker randomly samples a set of workers $\mathcal{S}_t^b$ in $\mathcal{S}_t$.
We let $X_{\mathcal{S}_t^b}$ and $\widehat{\theta} \left( X_{\mathcal{S}_t^b}  \right)$ denote the values provided by workers from set $\mathcal{S}_t^b$ and the estimator computed based on $X_{\mathcal{S}_t^b}$, respectively.
Then the attacker could approximate the distribution of $Q$ as follows:
\begin{align}
\widehat{Q}_b = \frac{\widehat{\theta} \left( X_{\mathcal{S}_t^b}  \right) - \widehat{\theta} \left( X_{\mathcal{S}_t}  \right)   }
{ \left[  \widehat{ \text{Var}}\left (  \widehat{\theta} \left( X_{\mathcal{S}_t^b}  \right) \right)  \right]   ^{1/2} /\sqrt{\vert \mathcal{S}_t \vert}}.
\label{Q_b_distribution}
\end{align}
Note that $\widehat{\theta} \left( X_{\mathcal{S}_t^b}  \right)$ and $\widehat{ \text{Var}}\left (  \widehat{\theta} \left( X_{\mathcal{S}_t^b}  \right) \right) $ can be computed on values $X_{\mathcal{S}_t^b}$. Theorem \ref{bootstrap_converge} states that $\widehat{Q}_b$ converges to $Q$ in distribution under the CRH model.
\begin{thm}\label{bootstrap_converge}
	Assume that $x_t^u \sim N\left( {x}^{\star}_t, {\sigma}_u^2 \right)$, where $u \in \mathcal{S}_t$.
	Let $Q$ and $Q^{\star}$ be defined as (\ref{Q_distribution}) and (\ref{Q_b_distribution}), respectively. 
	Then, for any real number $q$, we have that:
	\begin{align}
	\lim\limits_{\vert \mathcal{S}_t \vert \to \infty} \left\| \mathbb{P}^{\star}(\widehat{Q} \le q)  - \mathbb{P} (Q \le q)\right\| =0, \nonumber
	\end{align}
	where $\mathbb{P}^{\star}$ stands for the probability computed based on the bootstrapping sample distribution.
\end{thm}
\begin{proof}
	For the targeted item $t$, the set of workers whose values that can be observed by the attacker is $\mathcal{S}_t$, with the set size $\vert \mathcal{S}_t \vert$.
	Let $G_u {\left(\cdot \right)}$ denote the distribution of a sample $x_t^u$.
	Since $x_t^u$ follows from a normal distribution, we have $x_t^u \sim N\left( {x}^{\star}_t, {\sigma}_u^2 \right) = G_u {\left (x_t^u \right)}$.
	As shown in \cite{liu1988bootstrap,xiao2016towards}, we need to prove the following conditions to prove Theorem \ref{bootstrap_converge}:
	\begin{enumerate} [topsep=1pt, itemsep=-.1ex, leftmargin=.3in]
	\item[I) ] There exists a non-lattice distribution $H$ with mean zero and variance one, and a sequence $k_{\vert \mathcal{S}_t \vert}$ with $\frac{k_{\vert \mathcal{S}_t \vert}}{\log {\vert \mathcal{S}_t \vert}} \! \to \! \infty$, such that $k_{\vert \mathcal{S}_t \vert}$ of the population $G_u$, $u \in \mathcal{S}_t$, are of the form $G_u{\left(x \right)} = H{\left(\frac{x-\mu_u}{{\sigma}_u}  \right)}$ with ${\sigma}_u$ and $u \in \mathcal{S}_t$, bounded away from 0;
	\item[II) ] There exists an $M_1>0$ such that $\mathbb{E} \left[\vert x_t^u \vert ^{3+\delta _1} \right] \le M_1 < \infty$ for some $\delta _1 >0$;
	\item[III) ] $\liminf\limits_{\vert \mathcal{S}_t \vert \to \infty} \xi^2 >0$ and $\frac{1}{\vert \mathcal{S}_t \vert} \sum\nolimits_{u \in \mathcal{S}_t} {\left(\mu_u-\bar{\mu} \right)}^2 = o\left({\vert \mathcal{S}_t \vert}^{-1/2} \right) $;
	\item[IV) ] $H$ is continuous and there exists an $M_2>0$ such that for some $\delta_2 >0$, we have $ \mathbb{E} \left[\vert x_t^u \vert ^{6+\delta_2} \right] \le M_2 < \infty $,
	\end{enumerate}
	where $\mu_u = {x}^{\star}_t$, $\xi^2 = \frac{1}{\vert \mathcal{S}_t \vert}  \sum\nolimits_{u \in \mathcal{S}_t} {\sigma}_u^2$, $\bar{\mu}=\frac{1}{\vert \mathcal{S}_t \vert}  \sum\nolimits_{u \in \mathcal{S}_t} \mu_u$.

	\emph{Proof of I:}
	We let $H$ be the standard normal distribution, i.e., $H=N(0,1)$. 
	Then $H$ is a non-lattice distribution since any continuous distribution is non-lattice.
	If we let $k_{\vert \mathcal{S}_t \vert} = \vert \mathcal{S}_t \vert$ and $G_u{\left(x \right)} = H{\left(\frac{x-\mu_u}{{\sigma}_u}  \right)}$, then we have $\frac{\vert \mathcal{S}_t \vert}{\log {\vert \mathcal{S}_t \vert}} \to \infty $ as $\vert \mathcal{S}_t \vert$ increases.

	\emph{Proof of II:} According to the moments of a normal distribution, we have
	$
	\mathbb{E} \left[\vert x_t^u \vert ^k \right] = \sigma _u^{k}\frac{{{2^{\frac{k}{2}}}\Gamma \left( {\frac{{k + 1}}{2}} \right)}}{{\sqrt \pi  }}, \nonumber
	$
	where $k$ is any non-negative integer, $\Gamma(\cdot)$ is the gamma function.
	If we let $\delta _1 =1$, we have $\mathbb{E} \left[\vert x_t^u \vert ^4 \right] = \sigma _u^4\frac{{4\Gamma \left( {\frac{5}{2}} \right)}}{{\sqrt \pi  }} \stackrel{(a)}= 3\delta_u^4 = M_1 \le \infty$, where (a) follows from $\Gamma\left(\frac{5}{2}\right)=\frac{3}{4}\sqrt{\pi}$, which can be shown by the Legendre duplication formula that $\Gamma(z)\Gamma(z+\frac{1}{2})=2^{1-2z}\sqrt{\pi}\Gamma(2z)$ with $z=2$.

	\emph{Proof of III:} 
	${\sigma}_u^2>0$, $\mu_u = \bar{\mu}$, which completes the proof.

	\emph{Proof of IV:} Since $H=N(0,1)$, then $H$ is continuous.
	If we let $\delta_2 = 2$, then we have $\mathbb{E} \left[\vert x_t^u \vert ^{6+\delta_2} \right] = \sigma _u^8\frac{{2^{\frac{8}{2}} \Gamma \left( {\frac{9}{2}} \right)}}{{\sqrt \pi  }} = 105\delta_u^8 = M_2 < \infty$.

	Since all four conditions above are satisfied, we obtain:
	\begin{align}
	\mathbb{P} (Q \le q) = \Phi (q) + \frac{\beta_1}{6{\beta_2}\sqrt {\vert \mathcal{S}_t \vert} } \left(2q^2 +1 \right) \phi(q) + o\left({\vert \mathcal{S}_t \vert}^{-1/2} \right), \nonumber \\ 
	\mathbb{P}^{\star} (\widehat{Q} \le q) = \Phi (q) + \frac{\beta_3}{6{\beta_4}\sqrt {\vert \mathcal{S}_t \vert} } \left(2q^2 +1 \right) \phi(q) + o\left({\vert \mathcal{S}_t \vert}^{-1/2} \right), \nonumber
	\end{align}
	where $\Phi (q) = \frac{1}{{2\pi }}{\int_{ - \infty }^q e ^{ - {t^2}/2}}dt$,
	 $\phi(\cdot) $ is the derivative of $\Phi(\cdot) $
	(i.e., $\phi(\cdot) = \Phi^{\prime}  (\cdot)$),
	%
	$\beta_1 \!= \! \frac{1}{\vert \mathcal{S}_t \vert}  \sum\nolimits_{u \in \mathcal{S}_t} \! \mathbb{E}_{G_u} \! \left[{\left(x_t^u - \mu_u \right)}^3  \right] $,
	$\beta_2 \!= \! \frac{1}{\vert \mathcal{S}_t \vert}  \sum\nolimits_{u \in \mathcal{S}_t} \! {\sigma}_u^3$,
	$\beta_3 = \frac{1}{\vert \mathcal{S}_t \vert} \sum\nolimits_{u \in \mathcal{S}_t} {\left(x_t^u-  \rho \right)}^3 $,
	$\beta_4 = \zeta^3$,
	$ \rho = \frac{1}{\vert \mathcal{S}_t \vert}  \sum\nolimits_{u \in \mathcal{S}_t} x_t^u$,
	$\zeta=\sqrt{ \left(1/{\vert \mathcal{S}_t \vert}\right) \sum\nolimits_{u \in \mathcal{S}_t} {\left(x_t^u-  \rho \right)}^2 }$.
	Proof of II and III also show that $\beta_3 - \beta_1 \to 0$ as $\vert \mathcal{S}_t \vert \to \infty$, so we have $\mathbb{P}^{\star} (\widehat{Q} \le q) = \mathbb{P} (Q \le q) + O_p \left({\vert \mathcal{S}_t \vert}^{-1/2}  \right) $ 
	\footnote{
	 $X_n = O_p(Y_n)$ means ${X_n}/{\left\| {Y_n} \right\|}$ is bounded in probability, where $X_n$ and $Y_n$ are random sequences taking values in any normed vector spaces.}.
	 The proof is complete.
\end{proof}

%% file: experiments.tex

\section{Experiments} \label{sec:exp}

\subsection{Experimental Setup}

\subsubsection{Datasets} 
We first use a synthetic dataset to demonstrate the effectiveness of the proposed attack methods. 
In this synthetic dataset, there are 50,000 values in total on 4,000 items generated by 500 workers.
We assume that the value of worker $u$ on item $i$ follows a normal distribution $x_i^u \sim N\left( \mu_i, {\sigma}_u^2 \right)$, where $\mu_i$ is the ground truth of item $i$, ${\sigma}_u^2$ is the reliability of worker $u$. 
In the experiment, $\mu_i$ and ${\sigma}_u$ are generated from uniform distributions $\text{Uniform}(20, 30)$ and $\text{Uniform}(0,30)$, respectively.

To further demonstrate the advantages of the proposed attack methods, we also conduct experiments on two real-world continuous datasets, which are widely used for evaluating crowdsourcing systems. 
The first real-world dataset is \emph{Emotion}~\cite{snow2008cheap}, where the workers in this dataset need to assign a value from the interval [-100, 100] to some texts, indicating the degree of emotion (e.g., surprise) of the text. 
The second real-world dataset is \emph{Weather}~\cite{weatherdata_url}, which contains temperature forecast information for 88 major US cities collected from HAM weather \cite{AerisWeather_url}, Weather Underground (Wunderground) \cite{wunderground_url}, and World Weather Online (WWO) \cite{WorldWeatherOnline_url}.
The statistics of the three datasets are shown in Table~\ref{data_stat}.

\begin{table}[htbp]
	\centering
	\caption{Dataset statistics.}
	\label{data_stat}%
	\small
	\begin{tabular}{cccc}
		\hline
		Dataset & \# Workers & \# Items & \# Values \\
		\hline
		Synthetic & 500   & 4,000 & 50,000 \\
		Emotion & 38   & 700 & 7,000 \\
		Weather & 152   & 7,568 & 936,989 \\
		\hline
	\end{tabular}%
\end{table}%

\subsubsection{Attack Variants} We test and compare two variants of our proposed attack models, namely:

\myparatight{Full-knowledge attack} The attacker in this attack model is able to inject a set of malicious workers into the crowdsourcing systems. The attacker has full knowledge of the targeted system and sets values of the injected malicious workers according to Algorithm \ref{full_know_poisoning_attack_alg}.

\myparatight{Partial-knowledge attack} The attacker in this attack model is able to inject a set of malicious workers into the crowdsourcing systems. The attacker has partial knowledge of the targeted system and sets values of malicious workers according to Algorithm \ref{partial_know_poisoning_attack_alg}.

\subsubsection{Comparison of Attacks} To demonstrate the effectiveness of our proposed attacks, we compare our attack methods with the following methods.

\myparatight{Random attack} In this attack, for targeted item $t$, each malicious worker randomly assigns a number from the range $\left[ x_t^{\min}, x_t^{\max} \right]$ as the value for item $t$, where $x_t^{\min}$ and $x_t^{\max}$ are the minimum and maximum values on item $t$ provided by normal workers, respectively.

\myparatight{Maximum attack} In this attack model, for targeted item $t$, each malicious worker provides the maximum value $x_t^{\max}$ as the value for item $t$.

\subsubsection{Evaluation Metric} In order to measure the effectiveness of different attack models, we use the average estimation error defined in Eq.~(\ref{attack_goal}) as our evaluation metric.
Since the goal of the attack is to maximize the error of the aggregation results after attack, the larger the estimation error, the better the attack model.

\subsubsection{Parameter Setting}

Assume the attack size is $\alpha$ (i.e., the number of malicious workers is $\alpha$ fraction of the number of the total workers, $\alpha = \frac{\vert \mathcal{\widetilde{U}} \vert}{\vert \mathcal{U} \vert + \vert \mathcal{\widetilde{U}} \vert} $), and we can inject $\left \lfloor \frac{{\alpha \vert \mathcal{U} \vert }}{{1 - \alpha }} \right\rfloor$ malicious workers into the crowdsourcing systems. 
Then for a targeted item $i$, we random select $\left \lfloor \frac{{\alpha \vert \mathcal{U}_i \vert }}{{1 - \alpha }} \right \rfloor $ out of $\left \lfloor \frac{{\alpha \vert \mathcal{U} \vert }}{{1 - \alpha }} \right \rfloor$ malicious workers to attack item $i$.
In this setting, it is guaranteed that for each targeted item, the majority workers are normal workers.

Unless stated otherwise, we use the following default parameter setting:
We randomly select some items as targeted items and each targeted item is rated by at least 10 workers.
The numbers of targeted items are set to 400, 60 and 100 for Synthetic, Emotion and Weather datasets, respectively.
We let $B=500$.
We repeat each experiment for 50 trials and report the average
results.
All distance functions used in the experiments are squared distance.

\begin{figure}[!t]
	\centering
	\includegraphics[scale = 0.35]{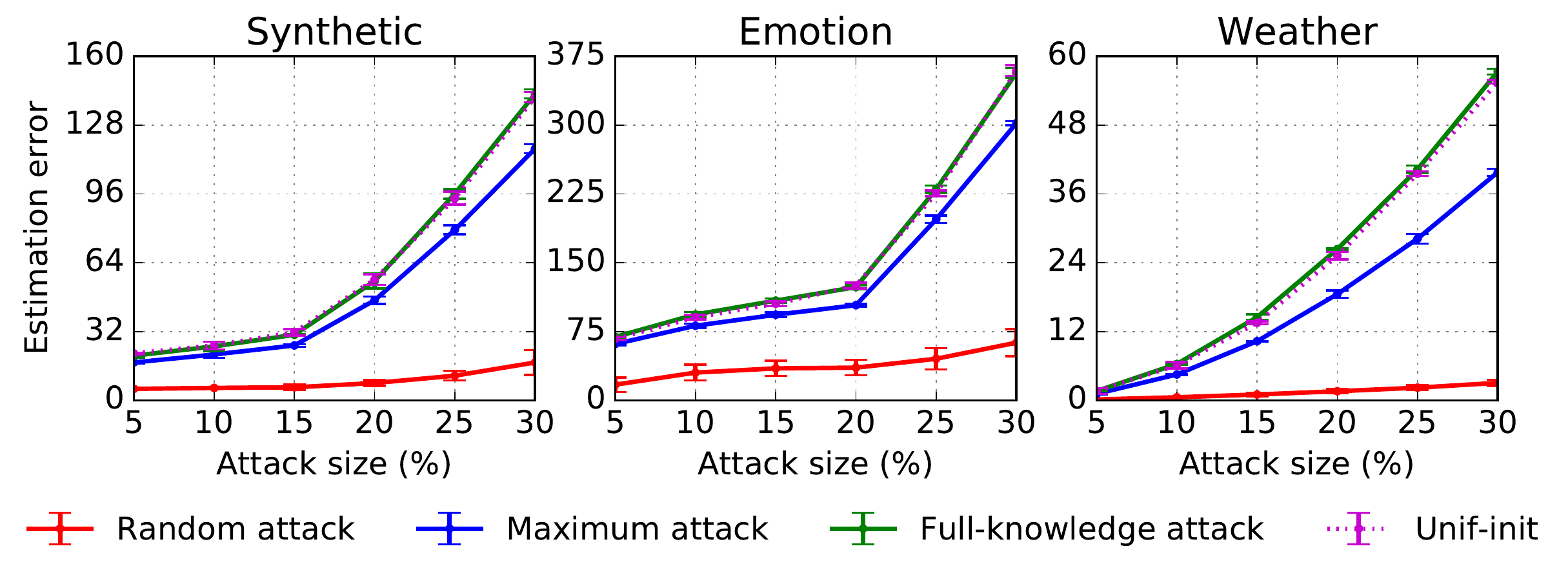}
	\caption{ Estimation error with respect to different attack sizes when attacking the CRH model.}\label{Utility_attack_size}
\end{figure}

\begin{figure}[!t]
	\centering
	\includegraphics[scale = 0.35]{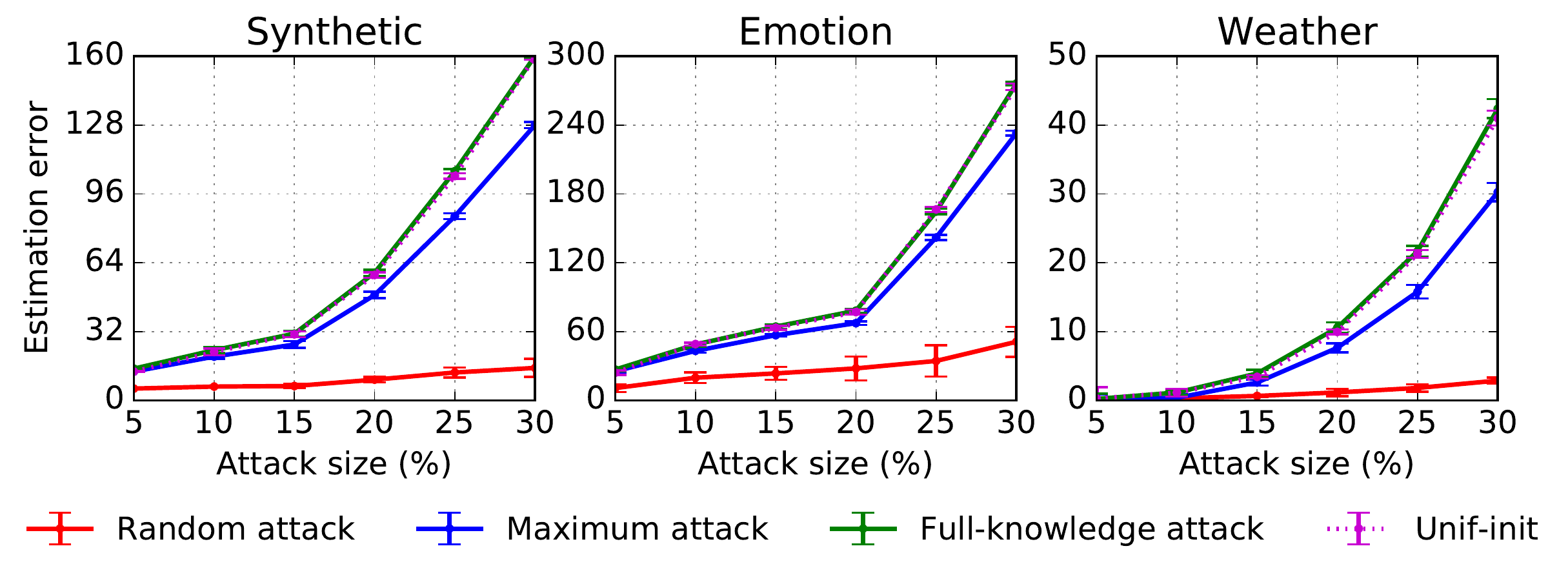}
	\caption{ Estimation error with respect to different attack sizes when attacking the GTM model.}\label{Utility_attack_size_GTM}
\end{figure}

\begin{figure}[!t]
	\centering
	\includegraphics[scale = 0.35]{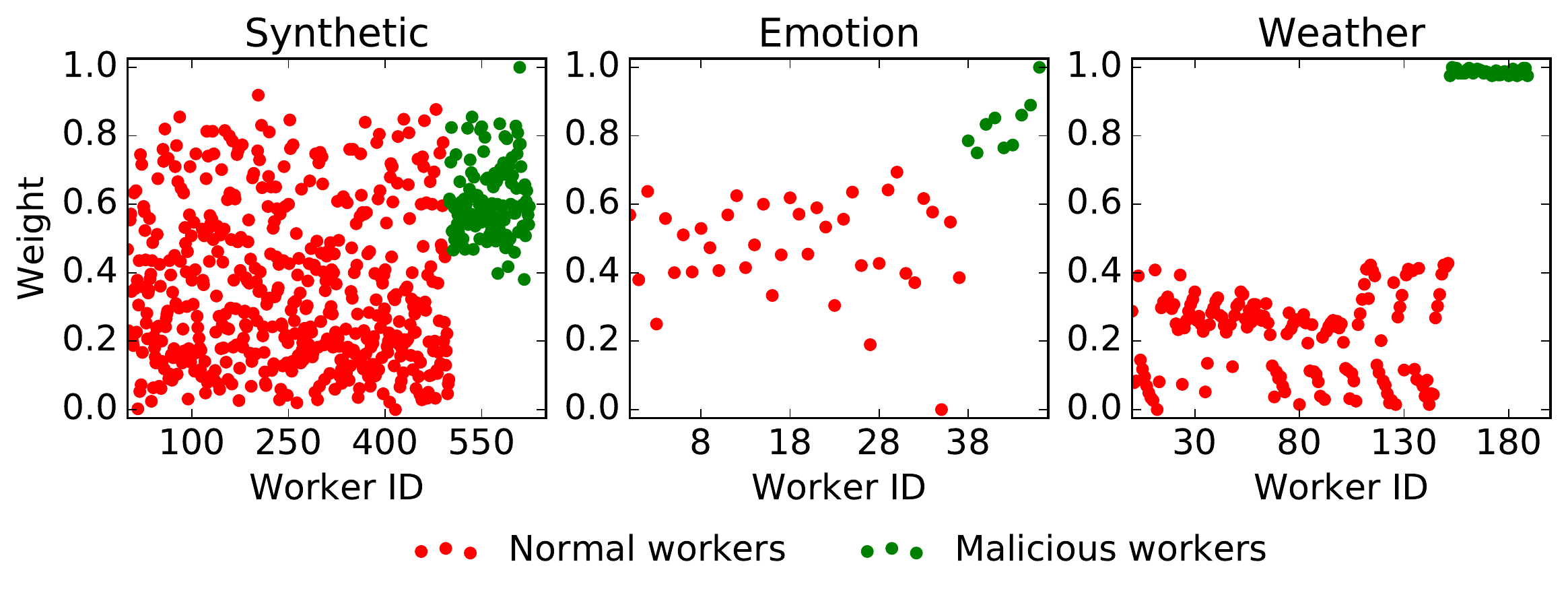}
	\caption{The weights of normal and malicious workers under maximum attack  when attacking the CRH model.}\label{max_weight}
\end{figure}

\begin{figure}[!t]
	\centering
	\includegraphics[scale = 0.35]{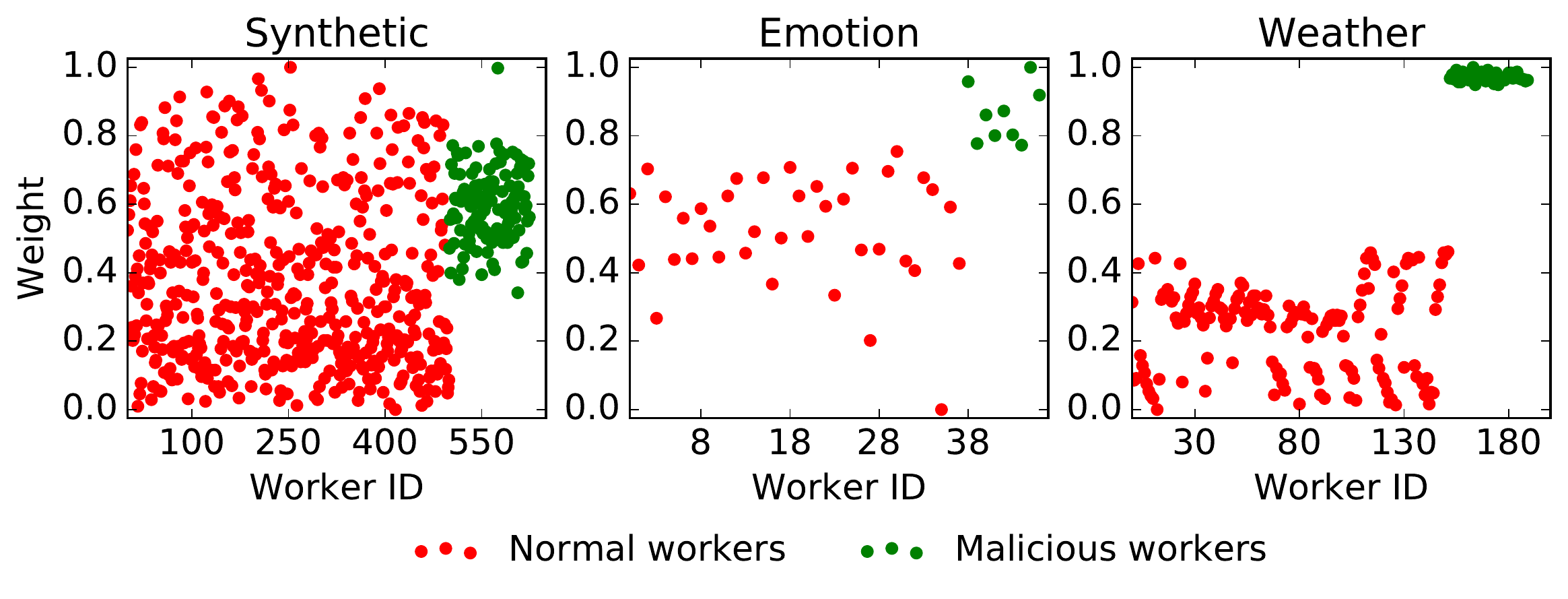}
	\caption{The weights of normal and malicious workers under full-knowledge attack when attacking the CRH model.}\label{full_weight}
\end{figure}

\begin{figure}[!t]
	\centering
	\includegraphics[scale = 0.35]{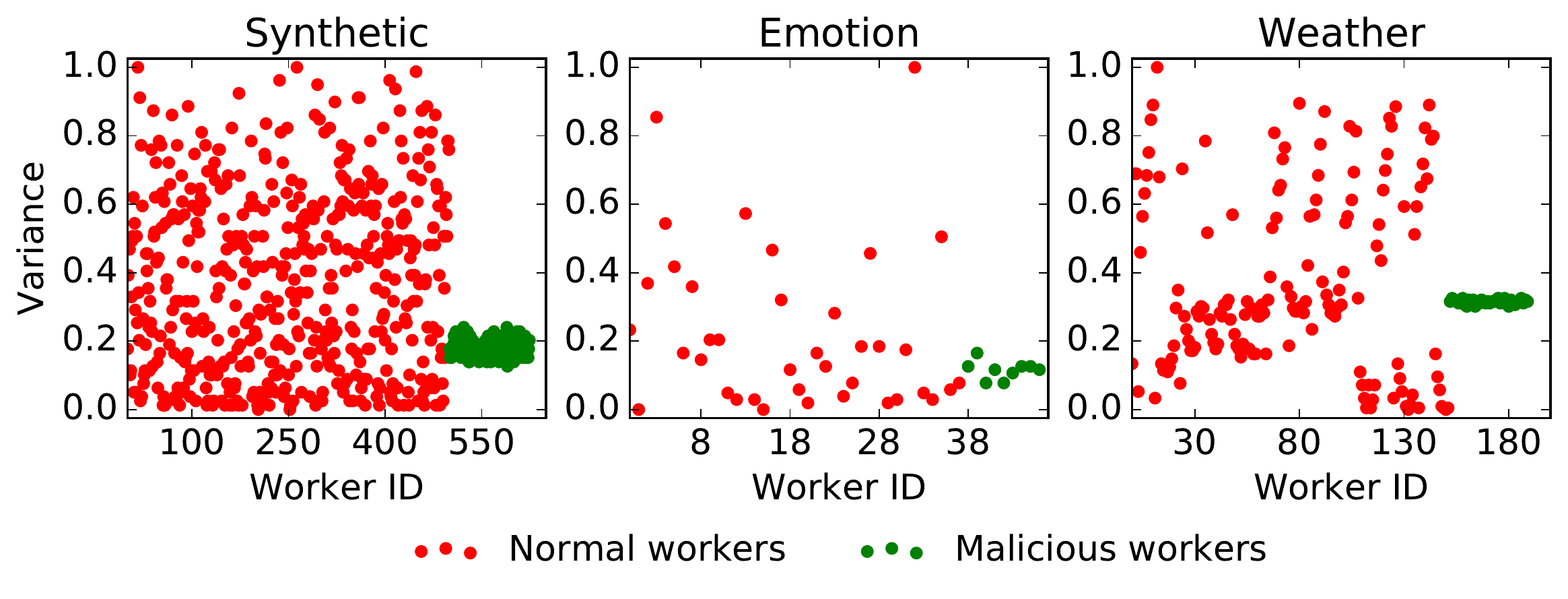}
	\caption{The variances of normal and malicious workers under maximum attack when attacking the GTM model.}\label{max_weight_GTM}
\end{figure}

\begin{figure}[!t]
	\centering
	\includegraphics[scale = 0.35]{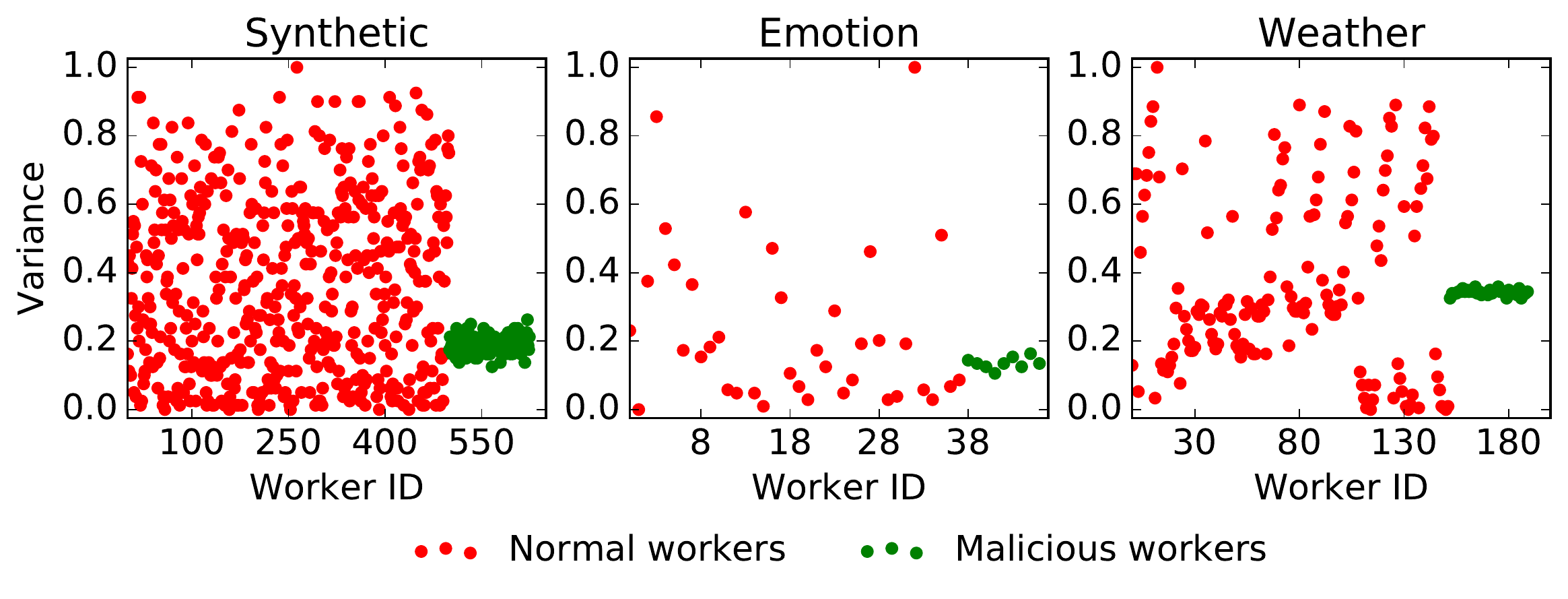}
	\caption{The variances of normal and malicious workers under full-knowledge attack when attacking the GTM model.}\label{full_weight_GTM}
\end{figure}

\subsection{Full-Knowledge Attack Evaluation}

For the optimization-based attack strategy, we first consider the full-knowledge attack, where the attacker knows the aggregation algorithm used in crowdsourcing systems (the CRH and GTM methods) and all values provided by normal workers.

\myparatight{Impacts of the attack size}
Figures~\ref{Utility_attack_size}-\ref{Utility_attack_size_GTM} show the average estimation errors of different attacks as the attack size (percentage of malicious workers) increases on three datasets, where the bar is standard deviation.
``Unif-init'' in Figures~\ref{Utility_attack_size}-\ref{Utility_attack_size_GTM} means for our full-knowledge attack algorithm, the server initializes the CRH/GTM model uniformly at random, and the initial weights/variances are drawn from the uniform distribution Uniform(2, 3). The attacker also initializes the attack model uniformly at random, but the initial weights/variances are drawn from the uniform distribution Uniform(1, 3).
%
First, we observe that our attack is effective in terms of inducing large estimation errors.
For instance, in the Emotion dataset, the attacker increases the estimation error to 93.69 by injecting 10\% of malicious workers for the CRH model.
Second, our proposed attack outperforms the baselines.
The reasons are as follows:
First, random and maximum attacks are general attack models and not optimized for CRH-based nor GTM-based truth discovery methods. 
Thus, their attack performance are not satisfactory.
Second, our proposed attack model takes the malicious workers' reliability into consideration. 
Specifically, our attack model abandons some targeted items when there is little chance to increase the aggregation error.
Thus, the malicious workers behave similarly with the majority of normal workers. 
By doing so, the crowdsourcing system may consider the malicious workers as normal workers and increase/decreases their weights/variances, which indirectly increases/decreases these malicious workers' weights/variances on other targeted items.
We also find that our attack increases the aggregation error significantly when we inject more malicious workers.
By contrast, random attack only slightly increases the estimation error. 
Another interesting finding is that even though the server and attacker adopt different ways to initialize the workers and not all workers’ initial weights/variances are equal, it does not affect the effectiveness of our proposed attack model.
From Figures~\ref{Utility_attack_size}-\ref{Utility_attack_size_GTM}, we observe that the standard deviations are very small, so we report the average results in the remaining experiments.

\myparatight{Comparisons between weights/variances of normal workers and malicious workers}
The CRH/GTM model uses the weights or variances to capture the workers’ quality. The key intuition of the CRH/GTM is that a worker should be assigned with a higher weight or lower variance if his values are closer to the estimated results.
In this experiment, we investigate the weight and variance distributions for both normal and malicious workers, the attack size is set to 20\%. 
As CRH and GTM models use different ways to measure the reliability of workers, we leverage min-max normalization technique to normalize reliability scores (weights and variances) into the range [0, 1].
The results are shown in Figures~\ref{max_weight}-\ref{full_weight_GTM}.
From Figure~\ref{full_weight} and Figure~\ref{full_weight_GTM}, we find that the malicious workers generated by our proposed full-knowledge attack have higher weights or smaller variances comparing with the normal workers.
This means that the malicious workers successfully blend into normal workers and it is hard to distinguish normal and malicious workers based on the weights/variances under our attack strategy.

\begin{figure}[!t]
	\centering
	\includegraphics[scale = 0.35]{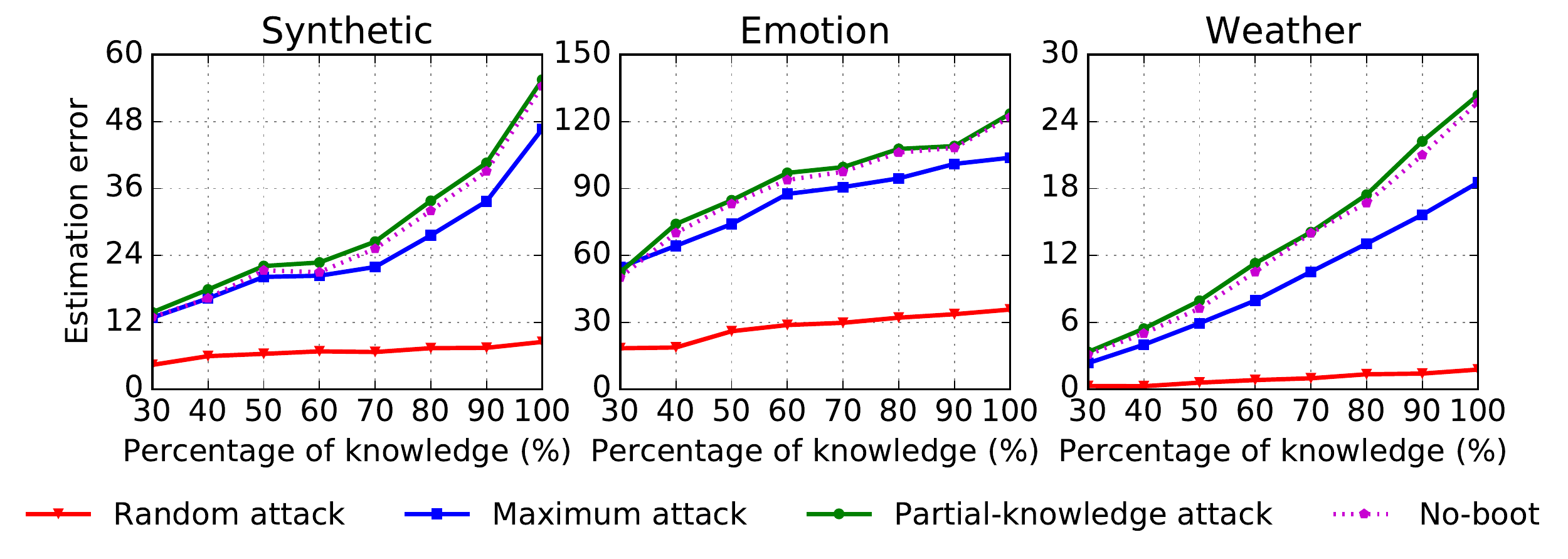}
	\caption{Estimation error with respect to the percentage of knowledge known by the attacker when attacking the CRH model.}\label{Utility_attack_partial}
\end{figure}

\begin{figure}[!t]
	\centering
	\includegraphics[scale = 0.35]{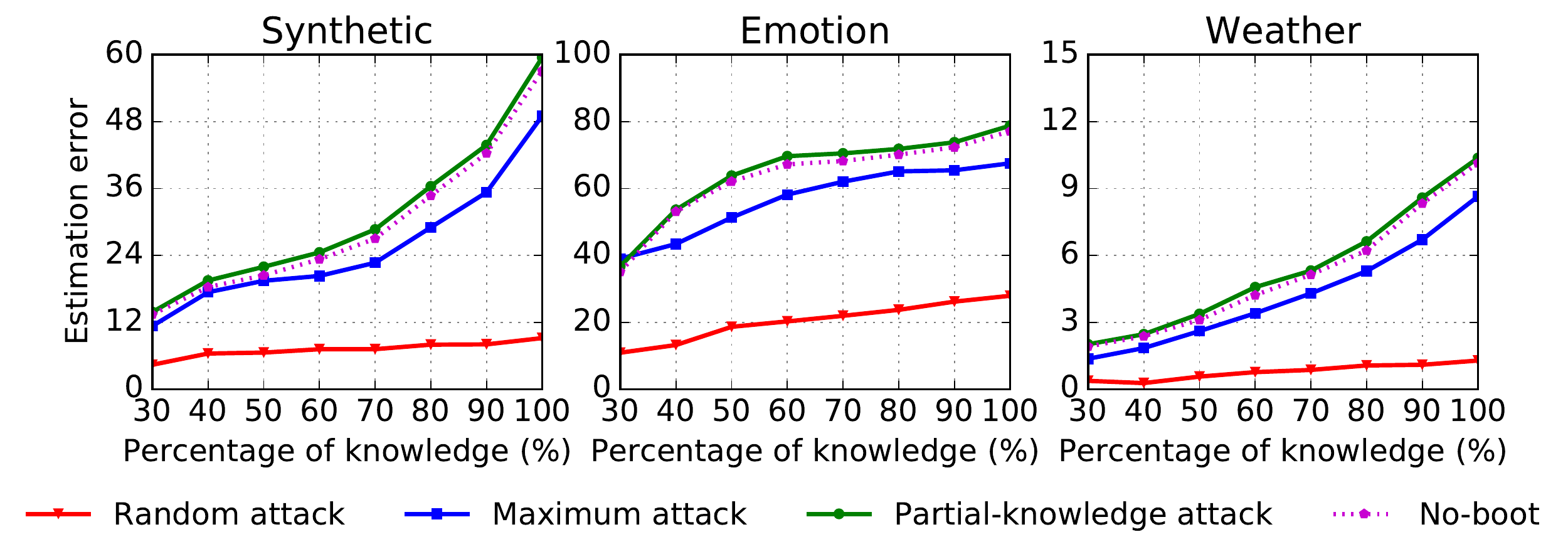}
	\caption{Estimation error with respect to the percentage of knowledge known by the attacker when attacking the GTM model.}\label{Utility_attack_partial_GTM}
\end{figure}

\subsection{Partial-Knowledge Attack Evaluation}

The amount of values can be accessed by the attacker is another important factor in the attack.
Figures~\ref{Utility_attack_partial}-\ref{Utility_attack_partial_GTM} show the results when the attacker only observes a portion of the values provided by normal workers on targeted items, where the attack size is set to be 20\%.
The percentage of knowledge in Figures~\ref{Utility_attack_partial}-\ref{Utility_attack_partial_GTM} represents the fraction of values provided by normal workers that can be observed by the attacker given a targeted item.
``No-boot'' means the attacker also sets values of malicious workers according to Algorithm \ref{partial_know_poisoning_attack_alg}. However, instead of leveraging the Bootstrapping technique to estimate the before-attack values, the attacker estimates the before-attack values using all observed values at once (without Bootstrapping).
Note that in the partial-knowledge attack, the attacker generates the values of malicious workers based only on the observed data.
We find that as the attacker has access to more data provided by normal workers, the estimation error increases (i.e., better attack performance).
We also find that our method achieves better attack performance than the baselines in most cases.
The reason is that our proposed partial-knowledge attack uses the Bootstrapping technique to combine estimated values from multiple bootstrapped values, rather than using all known values at once.
This leads to a more accurate value estimation that further slightly enhances the attack performance.

%% file: defense.tex

\section{Defenses} \label{sec:defense}

In this section, we propose two defense mechanisms to mitigate the impacts of poisoning attacks on crowdsourcing systems.
The basic idea in our defense mechanism design is to arm the crowdsourcing systems with malicious workers detection capability.

\subsection{Median-of-Weighted-Average Defense} 
\label{sec:MWA_defene}

Although the CRH and GTM models aim to provide robust aggregated results by assigning a larger weight or smaller variance to a worker if this worker's values are closer to the aggregated results, both the CRH and GTM models remain vulnerable to adversarial attacks.
To defend potential data poisoning attacks, we design a defense strategy that satisfies two goals:
1) similar to the CRH and GTM models, the server takes the quality of workers into account, and 2) the server should be resilient to potential poisoning attacks.
To achieve these goals, we propose the Median-of-Weighted-Average (MWA) defense.
In our MWA defense, the server is not aware whether the crowdsourcing system is being attacked.

Since GTM model can only handle continuous labels, while CRH model can deal with both categorical and continuous labels.
Thus, in the MWA defense, the server uses a weight parameter to capture a worker's reliability and updates the weights of workers the same way as the CRH model, i.e., weights are updated according to Eq.~(\ref{CRH_orgin_update_weight}). 
However, instead of updating the values based on Eq.~(\ref{CRH_orgin_update_truth}) directly, the server uses the following three steps to estimate the value for each item:
%
1) the server first sorts workers in ascending order according to the values provided by workers for this item, then partitions the workers (normal and malicious workers) who observe this item into $L$ groups;
%
2) the server then computes the weighted average of values in each group; and 3) the server takes the median of $L$ values as the estimated value for this item. 
For each item $i \in \mathcal{I}$, the estimated value $\widehat{x}_i^{*}$ can be computed as:
\begin{align}
\widehat{x}_i^{*}= \text{Median} 
\left(
\frac{{\sum\nolimits_{u \in {\mathcal{M}_i^1}} {{w_u}x_i^u}  } }
{{\sum\nolimits_{u \in {\mathcal{M}_i^1}} {w_u} }},...,
\frac{{\sum\nolimits_{u \in {\mathcal{M}_i^L}} {{w_u}x_i^u}  } }
{{\sum\nolimits_{u \in {\mathcal{M}_i^L}} {w_u} }}
\right),
\label{MWA_update_truth}
\end{align}
where $\mathcal{M} = \mathcal{U} \cup \mathcal{\widetilde{U}}$ is the set of all workers, $\mathcal{M}_i^l$, $l=1,...,L$, is the set of workers who observe item $i$ in the $l$-th group. 
The MWA defense is summarized in Algorithm~\ref{MWA-defense}.
In our proposed defense mechanisms, we also assume that all workers are given equal initial weights.

\begin{algorithm}[t!]
	\caption{{The Median-of-Weighted-Average (MWA) defense.}}
	\label{MWA-defense}
	\begin{algorithmic}[1]
		\renewcommand{\algorithmicrequire}{\textbf{Input:}}
		\renewcommand{\algorithmicensure}{\textbf{Output:}}
		\REQUIRE Values from all workers ${x_i^u}$ for $u \in \mathcal{M}, i \in \mathcal{I}$.
		\ENSURE  Aggregated values $X^{*}$ and worker weights $W$.
		\STATE Server initializes the workers’ weights.
		\WHILE {the convergence condition is not satisfied}
		\STATE For each item, the server partitions workers who observe this item into $L$ groups, then updates the aggregated value according to Eq.~(\ref{MWA_update_truth}).
		\STATE Server updates the weight of each worker according to Eq.~(\ref{CRH_orgin_update_weight}).
		\ENDWHILE
		\RETURN $X^{*}$ and $W$. 
	\end{algorithmic} 
\end{algorithm}

\subsection{Maximize Influence of Estimation Defense}

We note that, under the MWA defense, malicious workers still exist in the crowdsourcing systems.
In this section, we propose another defense mechanism to detect the malicious workers and remove them from the crowdsourcing systems.
However, this defense mechanism requires a \emph{stronger} assumption that the server knows the crowdsourcing system is being attacked and the goal of the attacker.
The server also knows there exists $\left \lfloor \alpha \vert \mathcal{M} \vert \right \rfloor$ number of malicious workers in the system, but the server does not know which items are being attacked.
Here, we propose the Maximize Influence of Estimation (MIE) defense to detect the malicious workers in the targeted systems.

\begin{algorithm}[!t]
	\caption{{Greedy influential worker selection.}}\label{Find_Influential_work_Set}
	\begin{algorithmic}[1]
		\renewcommand{\algorithmicrequire}{\textbf{Input:}}
		\renewcommand{\algorithmicensure}{\textbf{Output:}}
		\REQUIRE  Values from all workers ${x_i^u}$ for $u \in \mathcal{M}, i \in \mathcal{I}$.
		\ENSURE  Influential worker set $\mathcal{A}$.
		\STATE Initialize $\mathcal{A} = \emptyset $.
		\WHILE {$\mathcal{|A|} < \lfloor \alpha \vert \mathcal{M} \vert \rfloor$}
		\STATE Select $u = \arg {\max _{k \in \mathcal{M} \setminus \mathcal{A}}} \varphi(k,\mathcal{I})$.
		\STATE $\mathcal{A} \leftarrow \mathcal{A} \cup \{u\}$.
		\ENDWHILE
		\RETURN $\mathcal{A}$. 
	\end{algorithmic} 
\end{algorithm}

\begin{algorithm}[t!]
	\caption{The Maximize Influence of Estimation (MIE) defense.}
	\label{MIU-defense}
	\begin{algorithmic}[1]
		\renewcommand{\algorithmicrequire}{\textbf{Input:}}
		\renewcommand{\algorithmicensure}{\textbf{Output:}}
		\REQUIRE Values from all workers ${x_i^u}$ for $u \in \mathcal{M}, i \in \mathcal{I}$.
		\ENSURE  Aggregated values $X^{*}$ and worker weights $W$.
		\STATE Server initializes the workers’ weights.
		\STATE Server finds the influential worker set $\mathcal{A}$ according to Algorithm~\ref{Find_Influential_work_Set}.
		\STATE Server removes workers in the set $\mathcal{A}$ from the crowdsourcing systems.
		\WHILE {the convergence condition is not satisfied}
		\STATE Server updates the aggregated value of each item with the remaining workers according to Eq.~(\ref{CRH_orgin_update_truth}).
		\STATE Server updates the weights of the remaining workers according to Eq.~(\ref{CRH_orgin_update_weight}).
		\ENDWHILE
		\RETURN $X^{*}$ and $W$. 
	\end{algorithmic} 
\end{algorithm}

\begin{figure}[!t]
	\centering
	\includegraphics[scale = 0.35]{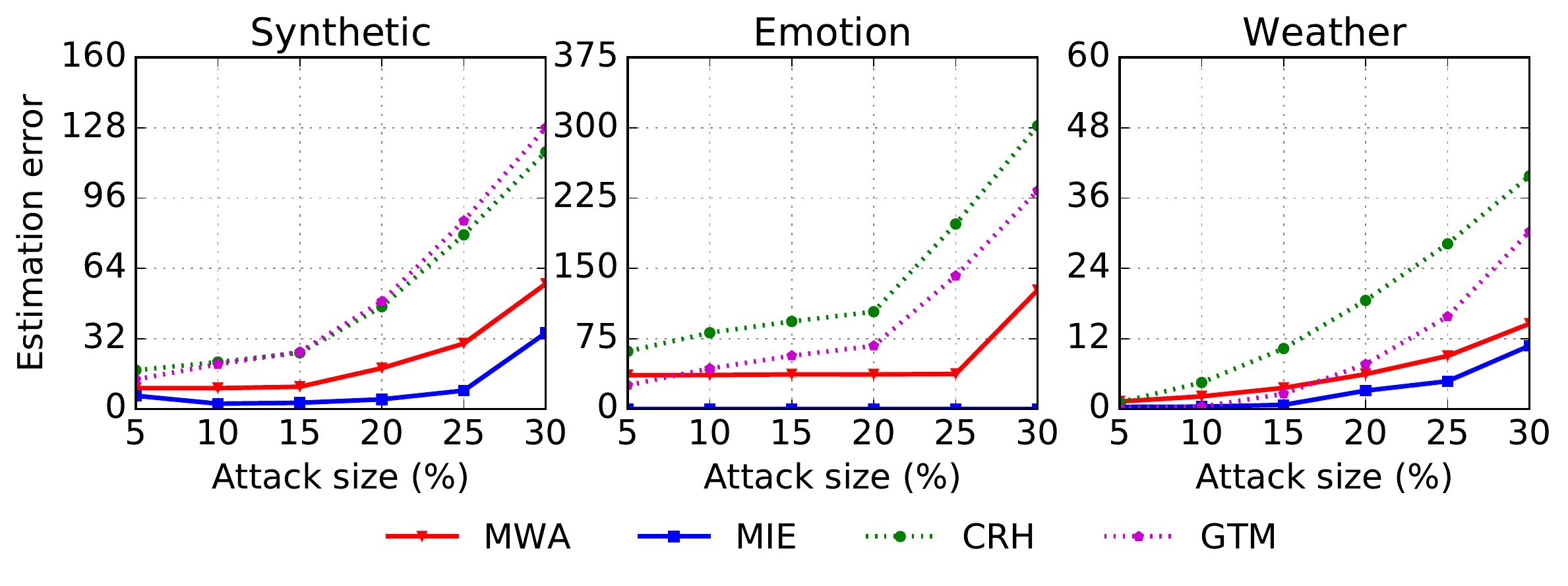}
	\caption{Estimation error of maximum attack when attacking the malicious-worker-aware crowdsourcing systems.}\label{defense_max_attack}
\end{figure}

\begin{figure}[!t]
	\centering
	\includegraphics[scale = 0.35]{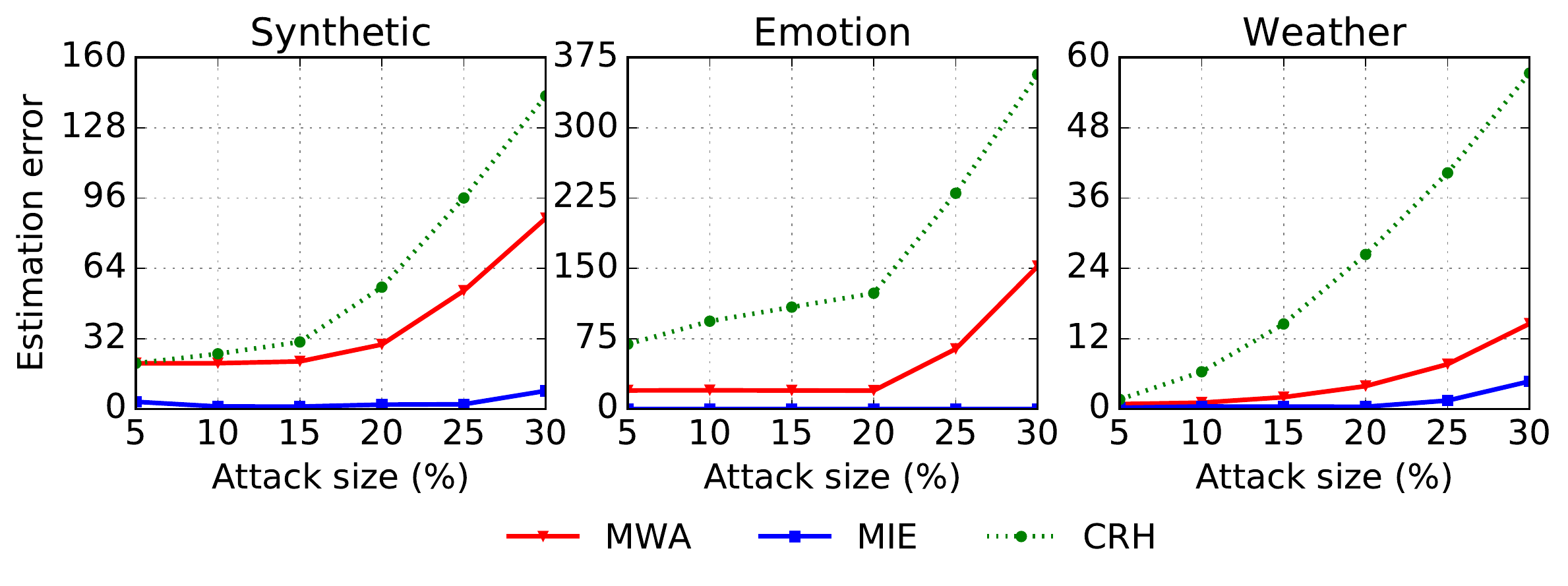}
	\caption{Estimation error of full-knowledge attack when attacking the malicious-worker-aware crowdsourcing systems, where the malicious workers are generated by the full-knowledge attack algorithm for the CRH model.}\label{defense_full_attack}
\end{figure}

\begin{figure}[!t]
	\centering
	\includegraphics[scale = 0.35]{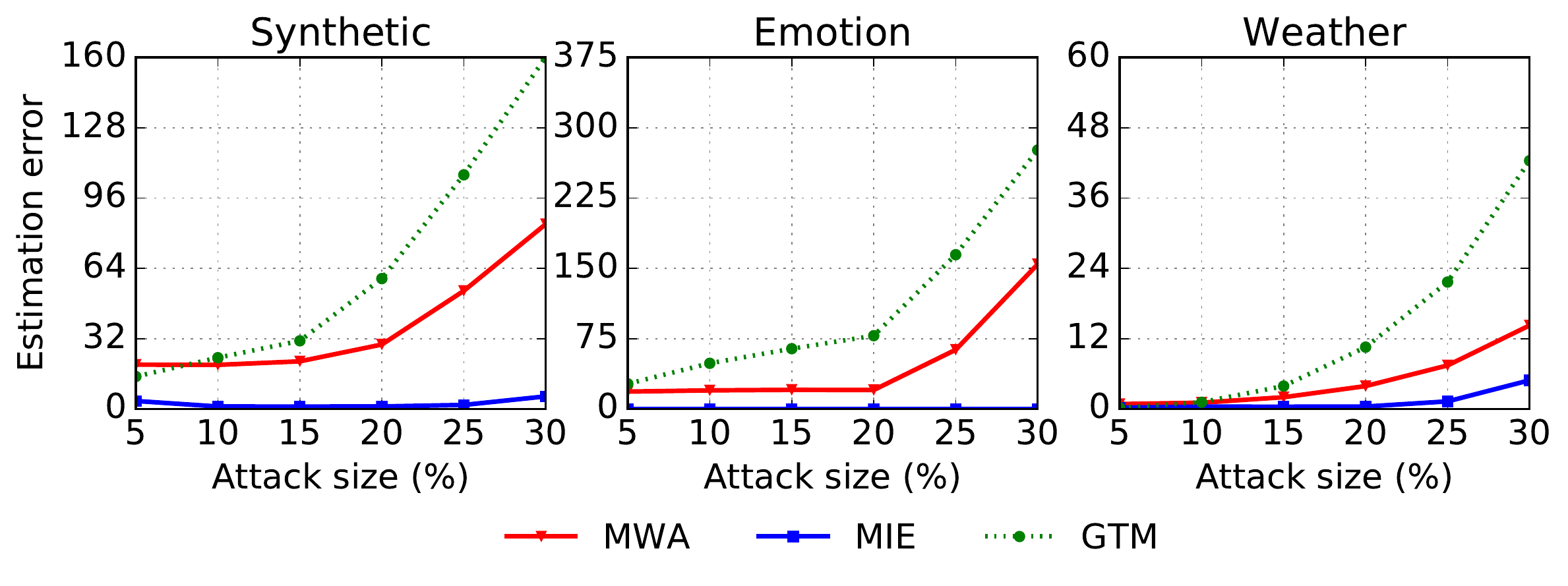}
	\caption{Estimation error of full-knowledge attack when attacking the malicious-worker-aware crowdsourcing systems, where the malicious workers are generated by the full-knowledge attack algorithm for the GTM model.}\label{defense_full_attack_GTM}
\end{figure}

For the MIE defense, we let $\mathbb{I}(\mathcal{A},\mathcal{I})$ denote the influence of removing workers in the set $\mathcal{A}$ on the estimation over all items in $\mathcal{I}$, where the influence here is defined as the change of estimated value. 
The server wants to find a set of influential workers that have the largest influence on all items in $\mathcal{I}$. 
The influence maximization defense problem can be formulated as:
\begin{align}
\label{influ_opt}
\text{Maximize}   \medspace\medspace \textstyle \mathbb{I}(\mathcal{A},\mathcal{I}),    \quad 
\text{subject to }   \medspace \vert {\mathcal{A}} \vert  =  \left \lfloor   \alpha \vert \mathcal{M} \vert \right \rfloor,
\end{align}

However, this combinatorial influence maximization problem is NP-hard \cite{kempe2003maximizing} in general. In order to solve Problem (\ref{influ_opt}), we first show how to quantify the influence of one worker, then we will show how to find a subset of $ \left \lfloor \alpha  \vert \mathcal{M} \vert \right \rfloor$ workers with the maximum influence.
We define $\varphi(u,\mathcal{I})$ as the influence of removing worker $u \in \mathcal{M}$ on the estimation over the targeted crowdsourcing system:
\begin{align}
\label{influence_user_target}
\varphi(u,\mathcal{I})  \stackrel{\text{def}} = \frac{1}{\left| \mathcal{I}_u \right|} \sum\nolimits_{i \in {\mathcal{I}}} 
d(\widehat{x}^{*}_i (\mathcal{M}),  \widehat{x}^{*}_i (\mathcal{M} \setminus \left\{ u \right\})),
\end{align}
where $\widehat{x}^{*}_i (\mathcal{M})$ represents the after-attack estimated value for item $i$ computed over all workers in set $\mathcal{M} = \mathcal{U} \cup \mathcal{\widetilde{U}}$, the distance function $d(\cdot)$ is squared distance, $\left| \mathcal{I}_u \right|$ is the number of items rated by worker $u$.
Therefore, the influence of removing workers from some set $\mathcal{A}$ on the estimation over the targeted system can be defined as the sum of the influence of individual worker in the set $\mathcal{A}$:
\begin{align}
\label{influ_set_S}
\mathbb{I}(\mathcal{A},\mathcal{I})  \stackrel{\text{def}} = \sum\nolimits_{u \in \mathcal{A}} {\varphi(u,\mathcal{I})}.
\end{align}

We can see that the set influence $\mathbb{I}(\mathcal{A},\mathcal{I})$ can be naturally computed based on the worker influence $\varphi(u,\mathcal{I}) $.
Note that even though the worker influence $\varphi(u,\mathcal{I})$ of Eq.~(\ref{influence_user_target}) shares some similarity with attacker's goal of Eq.~(\ref{attack_goal}), they have different meanings.
In Eq.~(\ref{attack_goal}), the attacker computes the average estimation errors of targeted items before and after attack, while the server in Eq.~(\ref{influence_user_target}) measures the change of estimated values of all items before and after removing one worker from the crowdsourcing systems.

\subsubsection{Approximation Algorithm to Determine $\mathcal{A}$} \label{sec:Approximating_influential_workers_set}

Although solving Problem (\ref{influ_opt}) is hard, we could design a greedy selection algorithm to approximately find a solution to Eq.~(\ref{influ_opt}) by leveraging the submodular property of influence $\mathbb{I}(\mathcal{A},\mathcal{I})$, which is stated in Theorem~\ref{monotonically_non_decreasing} as follows:
\begin{thm}\label{monotonically_non_decreasing}
	The influence $\mathbb{I}(\mathcal{A},\mathcal{I})$ is normalized, monotonically non-decreasing and submodular.
\end{thm}
\begin{proof}
	Define three sets $\mathcal{P}$, $\mathcal{K}$ and $\mathcal{Q}$, where $\mathcal{K} \subseteq \mathcal{P}$ and $\mathcal{Q} = \mathcal{P} \setminus \mathcal{K}$. 
	To simplify the notation, we use $\mathbb{I}(\mathcal{P})$ to denote $\mathbb{I}(\mathcal{P},\mathcal{I})$. 
	When there is no ambiguity, we let $\mathbb{I}(u)$ denote $\mathbb{I}(\{ u \})$ for $u \in \mathcal{M}$.
	Since $\mathbb{I}(\emptyset) = 0$, the influence function is normalized. We also have
	$
	\mathbb{I}(\mathcal{P}) - \mathbb{I}(\mathcal{K}) 
	= \sum\nolimits_{u \in \mathcal{P}} {\mathbb{I}(u)} - \sum\nolimits_{u \in \mathcal{K}} {\mathbb{I}(u)} 
	= \sum\nolimits_{u \in \mathcal{P} \setminus \mathcal{K} }  {\mathbb{I}(u)}
	= \mathbb{I}(\mathcal{Q})
	\ge 0 \nonumber,
	$
	which shows that influence $\mathbb{I}(\mathcal{A},\mathcal{I})$ is monotonically non-decreasing.
	To prove the submodular property, we define an arbitrary set $\mathcal{C}$ and we have
	$
	\mathbb{I}(\mathcal{P} \cup \mathcal{C}) - \mathbb{I}(\mathcal{K} \cup \mathcal{C})
	= \mathbb{I}((\mathcal{P} \cup \mathcal{C}) \setminus (\mathcal{K} \cup \mathcal{C})) 
	= \mathbb{I}(\mathcal{Q} \setminus (\mathcal{Q} \cap \mathcal{C})) \le \mathbb{I}(\mathcal{Q}) = \mathbb{I}(\mathcal{P}) - \mathbb{I}(\mathcal{K}). \nonumber
	$
	Thus the influence $\mathbb{I}(\mathcal{A}, \mathcal{I})$ is submodular and the proof is complete.
\end{proof}	

Based on the submodular property of influence $\mathbb{I}(\mathcal{A}, \mathcal{I})$, we propose a greedy selection method (Algorithm \ref{Find_Influential_work_Set}) to find an influential worker set $\mathcal{A}$ with $\left \lfloor \alpha \vert \mathcal{M} \vert \right \rfloor$ workers.
%
To be specific, we first compute the influence of each worker and add the worker with the largest influence to the set $\mathcal{A}$. Then we compute the influence of the remaining workers in the set $\mathcal{M} \setminus \mathcal{A}$, repeat this process until we find $\left \lfloor \alpha \vert \mathcal{M} \vert \right \rfloor$ workers.
Theorem~\ref{greedy_approximate} states that Algorithm \ref{Find_Influential_work_Set} finds a $(1-1/e)$ approximation solution with linear running time complexity.

\begin{thm} \label{greedy_approximate}
	Let $\mathcal{A}$ be an influential worker set returned by Algorithm~\ref{Find_Influential_work_Set} and $\mathcal{A^*}$ be the optimal influential worker set, respectively. It then holds that
	$\mathbb{I}(\mathcal{A},\mathcal{I}) \ge \left( {1 - \frac{1}{e}} \right) \mathbb{I}(\mathcal{A^*},\mathcal{I}) \nonumber$.
\end{thm}

\begin{proof}
	Let 
	$\mathcal{A}^* = \left\{
	a_1, a_2,...,a_{\left \lfloor \alpha \vert \mathcal{M} \vert \right \rfloor} 
	\right\}$ be the optimal influential worker set,
	 and $\mathcal{A}_i$ be the worker set after the $i$-th iteration of Algorithm~\ref{Find_Influential_work_Set}. 
	 Thus, we have
	$
	\mathbb{I}(\mathcal{A}^*) 
    \stackrel{(a)}  \le \mathbb{I}(\mathcal{A}_i \cup \mathcal{A}^*) 
	= \mathbb{I}(\mathcal{A}_i) + \mathbb{I}(\mathcal{A}_i \cup \{ a_1 \} ) - \mathbb{I}(\mathcal{A}_i) + 
	\mathbb{I}(\mathcal{A}_i \cup \{ a_1,a_2 \} ) - \mathbb{I}(\mathcal{A}_i \cup \{ a_1 \}) + \cdots 
	\stackrel{(b)} \le \mathbb{I}(\mathcal{A}_i) + \mathbb{I}(\mathcal{A}_i \cup \{ a_1 \} ) - \mathbb{I}(\mathcal{A}_i) + \mathbb{I}(\mathcal{A}_i \cup \{ a_2 \} ) - \mathbb{I}(\mathcal{A}_i) 
	+ \cdots + \mathbb{I}(\mathcal{A}_i \cup \{ a_{\left \lfloor \alpha \vert \mathcal{M} \vert \right \rfloor } \} ) - \mathbb{I}(\mathcal{A}_i) 
	\stackrel{(c)} \le \mathbb{I}(\mathcal{A}_i) + {\left \lfloor \alpha \vert \mathcal{M} \vert \right \rfloor}(\mathbb{I}(\mathcal{A}_{i+1})-\mathbb{I}(\mathcal{A}_{i})), 
	$
	where (a) follows from the monotonically non-decreasing property of influence $\mathbb{I}(\mathcal{A},t)$; 
	(b) uses the submodular property of influence; and (c) is due to $\vert {\mathcal{A}_i} \vert \le \left \lfloor \alpha \vert \mathcal{M} \vert \right \rfloor$.
	Arranging the terms, we obtain
	$
	\mathbb{I}(\mathcal{A}^*) -\mathbb{I}(\mathcal{A}_{i+1}) \le \left( {1 - \frac{1}{ \left \lfloor \alpha \vert \mathcal{M} \vert \right \rfloor}} \right)(\mathbb{I}(\mathcal{A}^*) - \mathbb{I}(\mathcal{A}_{i})). \nonumber
	$
	Recursively applying the inequality, we have
	$
	\mathbb{I}(\mathcal{A}^*) -\mathbb{I}(\mathcal{A}_{\left \lfloor \alpha \vert \mathcal{M} \vert \right \rfloor}) \le \left( {1 - \frac{1}{\left \lfloor \alpha \vert \mathcal{M} \vert \right \rfloor }} \right)^{\left \lfloor \alpha \vert \mathcal{M} \vert \right \rfloor }  \mathbb{I}(\mathcal{A}^*) \le \frac{1}{e} \mathbb{I}(\mathcal{A}^*).
	$
	Thus, we have $\left( 1- \frac{1}{e} \right) \mathbb{I}(\mathcal{A}^*) \le \mathbb{I}(\mathcal{A}_{\left \lfloor \alpha \vert \mathcal{M} \vert \right \rfloor })$,
	which completes the proof.
\end{proof}

After using the influence function $\mathbb{I}(\mathcal{A},\mathcal{I})$ to find the influential worker set $\mathcal{A}$ with $\left \lfloor \alpha \vert \mathcal{M} \vert \right \rfloor$ workers, the server then removes workers in set $\mathcal{A}$ from the crowdsourcing systems (the server views these workers as malicious workers), and finally estimates the value for each item with the remaining workers.
In our MIE defense, we use the CRH model to find influential workers and estimate the aggregated values of items.
Our MIE defense is stated in Algorithm~\ref{MIU-defense}.

\subsection{Defense Evaluation}

Figures~\ref{defense_max_attack}-\ref{defense_full_attack_GTM} show the average estimation errors of different attacks on the CRH, GTM, MWA and MIE methods, where full-knowledge attacks are considered.
The number of groups in MWA defense is set to 5, 4 and 5 for Synthetic, Emotion and Weather datasets, respectively.
From Figures~\ref{defense_max_attack}-\ref{defense_full_attack_GTM}, we observe that both MWA and MIE defenses could mitigate the impacts of malicious workers.
MIE achieves a better defense performance compared to MWA.
However, MIE and MWA are \emph{not} directly comparable since we assume the server knows the crowdsourcing system is being attacked and the server knows the number of malicious workers exist in the system in MIE.
For the MWA defense mechanism, the strategy of dividing workers into groups and computing the median between different groups can only reduce the impact of malicious workers since malicious workers still exist in the system; 
and if the percentage of malicious workers is high, there would be more malicious workers in each group on average, which leads to less robust weighted average in each group.
We also find that even if the server is equipped with malicious workers detection capability, our proposed MWA and MIE defenses may still be vulnerable to poisoning attacks if the percentage of malicious workers is high.
For example, the average estimation error of MWA is still 14.57 on the Weather dataset when the attacker injects 30\% of malicious workers under the CRH model.

%% file: conclusion.tex

\section{Conclusion} \label{sec:conclusion}

In this paper, we performed a systematic study on data poisoning attacks and defenses to crowdsourcing systems.
We demonstrated that crowdsourcing systems are vulnerable to data poisoning attacks.
We proposed an optimization-based data poisoning attack to blend malicious workers into normal workers and increase the estimation errors of the aggregated values for attacker-chosen targeted items.
Our attacks are effective under both full-knowledge and partial-knowledge settings.
Furthermore, we designed two defense mechanisms to mitigate the impacts of malicious workers.
Our results showed that our proposed attacks can increase the estimation errors substantially and our defenses are effective.


\section*{Acknowledgements}
This work is supported in part by NSF grants CAREER CNS-2110259, CIF-2110252, ECCS-1818791, CCF-2110252, CNS-1937786, IIS-2007941, ONR grant ONR N00014-17-1-2417, and a Google Faculty Research Award.